\documentclass[a4paper,UKenglish,cleveref,autoref,thm-restate,nolineno]{socg-lipics-v2021}

\pdfoutput=1 %
\hideLIPIcs

\usepackage{graphicx}
\usepackage[table]{xcolor}

\usepackage{amsmath}
\usepackage{amssymb}
\usepackage{mathtools}
\usepackage{nicefrac}
\usepackage{microtype}
\usepackage{todonotes}
\usepackage{cite}

\theoremstyle{definition}%
\newtheorem{ilpdef}{ILP}%



\usepackage{csquotes} %
\usepackage{marvosym} %

\colorlet{mygreen}{green!70!black}
\colorlet{myblue}{cyan!90!black}
\colorlet{myorange}{orange!75!white}

\usepackage{booktabs} %
\usepackage{multirow} %
\usepackage{array}    %
\usepackage{calc}     %

\usepackage{xurl} %
\Crefname{lemma}{Lemma}{Lemmas}
\Crefname{observation}{Observation}{Observations}
\Crefname{corollary}{Corollary}{Corollaries}
\Crefname{definition}{Definition}{Definitions}
\Crefname{case}{Case}{Cases}
\Crefname{ilpdef}{ILP}{ILPs}
\crefalias{subequations}{ilp}%
\Crefname{ilp}{ILP}{ILPs}
\crefname{ilp}{ILP}{ILPs}

\usepackage[abbreviations,shortcuts]{glossaries-extra}

\newglossary*{problem}{Problems}
\newcommand{\newproblem}[3]{
    \newglossaryentry{pr:#1}{type=problem,
        name={\ensuremath{#2}},
        description={#3},
        sort={#1}
    }
}
\newcommand{\prob}[1]{\glsentryname{pr:#1}}
\newcommand{\probf}[1]{\glsentrydesc{pr:#1} (\glsentryname{pr:#1})}
\newcommand{\probl}[1]{\glsentrydesc{pr:#1}}

\newabbreviation{DAG}{DAG}{directed acyclic graph}
\newabbreviation{MCPS}{CPS}{minimum capacity-preserving subgraph}
\newabbreviation{ISP}{ISP}{Internet service provider}
\newabbreviation{ECMP}{ECMP}{Equal Cost Multipath}
\newabbreviation{IGP}{IGP}{Interior Gateway Protocol}
\newabbreviation{MLU}{MLU}{maximum link utilization}
\newabbreviation{MPLS}{MPLS}{Multiprotocol Label Switching}
\newabbreviation{OSPF}{OSPF}{Open Shortest Path First}
\newabbreviation{RSVP}{RSVP}{Resource Reservation Protocol}
\newabbreviation{TE}{TE}{traffic engineering}
\newabbreviation{DFS}{DFS}{depth-first search}
\newabbreviation{CSP}{CSP}{constrained shortest path}

\newproblem{MCPS}{\textsf{CPS}}{\textsf{Capacity-Preserving Subgraph}}
\newproblem{MSPND}{\textsf{SPND}}{\textsf{Shortest-Path Network Design}}
\newproblem{TOCA}{\textsf{TOCA}}{\textsf{Traffic-Oblivious Cable Activation}}
\newproblem{DSND}{\textsf{DSND}}{\textsf{Directed Survivable Network Design}}
\newproblem{MSRND}{\textsf{SRND}}{\textsf{Segment-Routing Network Design}}
\newproblem{F-}{\textsf{F-}}{\textsf{Fixed-Shortest-Paths}}
\newproblem{FCND}{\textsf{FCND}}{\textsf{Fixed-Charge Network Design}}

\let\Union\bigcup

\renewcommand{\epsilon}{\varepsilon}
\renewcommand{\rho}{\varrho}
\newcommand{\bigO}{\ensuremath{\mathcal{O}}}

\newcommand{\Npos}{\ensuremath{\mathbb{N}_{\geq1}}}

\newcommand{\be}{\ensuremath{\coloneqq}}
\newcommand{\ssep}{\mid}

\newcommand{\aarc}{\ensuremath{a}}
\newcommand{\arcs}{\ensuremath{A}}

\newcommand{\pathset}{\ensuremath{\mathcal{P}}}

\DeclarePairedDelimiter\ceil{\lceil}{\rceil}

\newcommand{\ecap}{\ensuremath{\textit{fcap}}}
\newcommand{\maxflow}[1]{\ensuremath{c_{#1}}}
\newcommand{\stretcha}{\varrho}

\newcommand{\mcfalpha}{\varrho}

\newcommand{\demand}{\ensuremath{T}}

\newcommand{\paredges}[1]{\ensuremath{\lambda_{#1}}}

\newcommand{\minpar}{\ensuremath{\paredges{\text{min}}}}

\newcommand{\ccap}{\ensuremath{\textit{ccap}}}

\newcommand{\len}{\ensuremath{\textit{len}}}
\newcommand{\ppath}{\ensuremath{P}}
\newcommand{\stpaths}{\ensuremath{\pathset_{s,t}}}
\newcommand{\termpairs}{\ensuremath{K}}
\newcommand{\pless}{\ensuremath{\prec}}
\newcommand{\subpath}[2]{\ensuremath{#1\langle#2\rangle}}
\newcommand{\numconns}{\ensuremath{\chi}}

\newcommand\algspnd{\emph{spnd}}
\newcommand\algspndone{\ensuremath{\emph{spnd}^{\,\circ}}} %
\newcommand\algcps{\emph{cps}}
\newcommand\algmcf{\emph{toca}}
\newcommand\algmcfpp{\ensuremath{\emph{toca}^+}}
\newcommand{\cost}{\ensuremath{\textit{dcost}}}
\newcommand{\effcap}{\ensuremath{\textit{ecap}}}

\usepackage{tikz}
\usetikzlibrary{arrows,
    backgrounds,
    shapes,
    calc,
    decorations,
    decorations.pathreplacing,
    decorations.pathmorphing,
    positioning}
\tikzset{main_edge/.style={line width=2.3pt}}
\tikzset{dot/.style={circle,fill=black,inner sep=1pt,minimum size=1pt}}
\tikzset{mylabel/.style={
    fill opacity=0, draw opacity=0, text opacity=1, inner sep=0, outer sep=3.5pt}}
\tikzset{edge/.style={
    rounded corners, -stealth, draw=black, very thick, shorten <= 2pt,shorten >= 2pt}}

\title{Designing Capacitated Subnetworks for Shortest~Path~Routing}

\author{Markus Chimani}{Theoretical Computer Science, Osnabrück University,
Osnabrück, Germany}{markus.chimani@uos.de}{https://orcid.org/0000-0002-4681-5550}{}%

\author{Max Ilsen\footnote{Corresponding author.}}{Theoretical Computer Science, Osnabrück University,
Osnabrück, Germany}{max.ilsen@uos.de}{https://orcid.org/0000-0002-4532-3829}{}

\authorrunning{M. Chimani and M. Ilsen} %

\Copyright{Markus Chimani and Max Ilsen} %



\category{} %

\relatedversion{} %

\funding{Funded by the Deutsche Forschungsgemeinschaft (DFG, German Research Foundation) grant 461207633 (CH 897/7-2).}

\begin{document}
\maketitle

\begin{abstract}
    In pursuit of higher energy efficiency in computer networks,
    one subfield of \emph{green traffic engineering} aims at reducing the size
    of a network during times of low traffic, while still guaranteeing the
    ability to route all occurring demands. In this setting, we  have to simultaneously solve a network design problem (choosing connections to deactivate) and a routing problem (routing paths in the active subnetwork, adhering to some routing protocol).
    Interestingly, there seems to be no available method to tackle the problem as a whole for the simplest (and still most commonly used) routing paradigm: shortest path routing. State-of-the-art methods either do not consider capacities, or assume that the routing paths should not change when deactivating network connections, or separate the problem into its two constituents, first solving the network design problem (using some estimators in lieu of the precise routing protocol) and only then the actual routing problem.

    In this paper, we present an algorithm to tackle the full combined problem exactly via a novel integer linear program, modeling dynamically changing shortest paths. To solve it, we need to devise a special-purpose column generation method. To speed up the solution process, we further propose additional provably strengthening constraints.

    Now having the means to yield true optimal solutions for (small) practical instances, we can for the first time give an in-depth experimental evaluation that includes the absolute quality (or quality loss) intrinsic to the above simplifying algorithms.
    It turns out that the arguably simplest method---first computing a routing, fixing it, and then turning off all superfluous connections---yields solutions surprisingly close to the true optimum in practice.
    When considering multiple different traffic demands, a recent traffic-oblivious approach (TOCA) performs best, while being comparatively straight-forward to implement.

\end{abstract}

\newpage
\section{Introduction}\label{sec:intro}

Data traffic in modern backbone networks of \acp{ISP} showcases distinct
peaks during the day and evenings, and lows at night and in the early
morning~\cite{DBLP:conf/lcn/SchullerACHS17,DBLP:conf/infocom/HassidimRSS13}.
Green traffic engineering approaches exploit this by turning off unnecessary
resources during traffic lows, thereby saving energy and money.
In this paper, we examine the emerging network design problems w.r.t.\ two different network models:
in the setting of simplex (full-duplex) communication,
we model a network as a (bi)directed graph~$G=(V,\arcs)$ with positive (symmetric) arc capacities~$\ccap \colon \arcs\to\Npos$ and arc lengths~$\len \colon \arcs\to\Npos$, where each vertex corresponds to a router
and each arc (each pair of two opposite-facing arcs, respectively) to a link.
Each link with a corresponding arc~$\aarc\in\arcs$ is understood to consist of~$\paredges{\aarc} \in \Npos$ many \emph{connections}, each of which provides a capacity of~$\ccap$ and is capable of being deactivated separately. The full link, with all connections activated, thus has a capacity of $\ecap(\aarc)\coloneqq\paredges{\aarc}\cdot\ccap(\aarc)$.

A \emph{traffic demand} of a vertex pair $(s,t)\in V^2$ is the amount of data (\emph{flow} in standard algorithmic terms) that needs to be sent from $s$ to $t$. All these demands are encoded in a (non-symmetric) \emph{traffic matrix} $\demand\in\mathbb{N}^{|V|\times|V|}$.
The \emph{terminal pairs} $\termpairs \be \{(s,t) \in V^2 \ssep \demand(s,t) > 0\}$ are those vertex pairs that have a positive associated
traffic demand.
There are several routing strategies:
\emph{shortest-path routing (SPR)}, codified as \ac{OSPF}, is the most commonly employed \ac{IGP} in
real-world traffic engineering scenarios~\cite{DBLP:journals/rfc/rfc2328}. There, traffic is always sent along a shortest path w.r.t.\ the arc lengths~$\len$, which are usually called \emph{\ac{IGP} weights} or \emph{link metric}.
The newer, not yet as widely adopted,
\emph{segment routing (SR)}~\cite{DBLP:journals/rfc/rfc8402} chooses intermediary vertices (in many cases only one per traffic demand~\cite{DBLP:conf/infocom/BhatiaHKL15})
and routes traffic by concatenating shortest paths through them.
In general, routing over arbitrary paths via \emph{multi-commodity flow (MCF)} is not feasible due to
hardware constraints~\cite{DBLP:journals/ojcs/OttenICA26}.
We say~$\demand$ is \emph{SPR-routable} (\emph{SR-routable}, \emph{MCF-routable}) in $G$ if the respective routing protocol is able to satisfy all traffic demands without violating the capacities in~$G$.
In contrast to the other two, already testing whether~$\demand$ is SR-routable in $G$ is NP-hard~\cite{DBLP:conf/cp/HartertSVB15}.

Given a network~$(G,\ccap,\len,\paredges{})$ and a traffic matrix~$\demand$,
the \probf{MSPND} %
problem consists of finding a vector $\numconns\in\mathbb{N}^{|\arcs|}$ with $\numconns_\aarc\leq\paredges{\aarc}$ for all $\aarc\in\arcs$ and minimum %
$\sum_{\aarc\in\arcs}\numconns_\aarc$ such that $\demand$ is SPR-routable in the \emph{subnetwork} $(G,\ccap,\len,\numconns)$. Intuitively, we say link $\aarc\in\arcs$ has $\numconns_\aarc$ \emph{active} connections, and aim to minimize the total number of active connections in the network.
We may analogously define the \probf{MSRND} problem w.r.t.\ SR-routability.

The established problem probably most similar to \prob{MSPND} is \probl{FCND}~\cite{crainic2021networkdesign}, where there are some known algorithms based on bi-level programs when combined with shortest path routing~\cite{DBLP:journals/cor/BourasFPZ19,DBLP:journals/cor/SilvaSMMS16}. However, there, no capacities are enforced on any arcs, but the objective function sums fixed arc costs for activating an arc plus arc costs per unit of traffic routed over an arc. As such, these results are not applicable to the \prob{MSPND} scenario.

To the best of our knowledge, there are no known approaches that directly tackle the true \prob{MSPND} or \prob{MSRND} problems, as the combination of deciding on a capacitated subnetwork simultaneously with solving the routing problem is quite intricate. We want to emphasize the major difficulty intrinsic already to shortest paths routings due to their non-monotonicity in subnetworks: Assume some traffic matrix~$\demand$ is \emph{not} SPR-routable in some network;
by deactivating a single link, even though
this decreases the overall capacity of the network, $\demand$ may become indeed SPR-routable, as traffic demands
would now be routed over different paths with possibly higher residual capacity.
This phenomenon is also known as Braess’s~paradox~\cite{DBLP:journals/mmor/Braess68}.

In practice, the problems are typically considered in simplified forms:
The most straightforward approach is to stipulate that the shortest routing paths are not allowed to change in the subnetwork. We call this the \probl{F-} problem variant, marked by the prefix \prob{F-}.
Then \prob{F-}\prob{MSPND} becomes trivial as one simply performs a standard SPR-routing and deactivates superfluous connections. In contrast, \prob{F-}\prob{MSRND} still is NP-hard even with this stipulation, but there is an ILP-based exact algorithm~\cite{DBLP:conf/lcn/OttenBSA23} that can be used in practice for reasonably sized graphs. However, requiring to stick to the shortest-path routes of the full network effectively means that no link used by the initial routing can ever be completely deactivated, and much potential for energy saving seems wasted. Indeed, the optimal solution to \prob{F-}\prob{MSPND} can be a factor of~$\Theta(|V|)$ larger than an optimal solution to \prob{MSPND}:
consider a network on a complete digraph with unit arc lengths and large capacities per connection; while for the latter a spanning cycle on $\bigO(|V|)$ arcs, with one connection each, may suffice, the former may have to retain at least one connection for all $\bigO(|V|^2)$~many original links.

The alternative idea is to split the problem into two subproblems that are solved sequentially: first, one establishes an SPR- or SR-routable subnetwork, then one computes the actual routing. The crux is, of course, in order to guarantee routability in the first step, one would already need to know the routing of step two. The key idea is to consider \emph{traffic-oblivious} schemes in the first step.

A \emph{traffic-aware} approach requires the knowledge of a specific traffic matrix that should be routed in the subnetwork. This approach has the drawback of overfitting for this specific traffic matrix, when in reality, one only has (rough) estimates of typical traffic demands.
\emph{Traffic-oblivious} approaches mitigate this issue and simultaneously allow us to algorithmically split \prob{MSPND} (and \prob{MSRND}) into two phases as sketched above. Such approaches are based on the underlying assumptions that (a) the original network (i.e., with all connections active) is well-designed to route the traffic at peak hours, and (b) the traffic at low hours is, mostly, a scaled-down traffic of that at peak hours. The idea is to compute the subnetwork using only the estimated scale factor~$\stretcha\in(0,1)$, but no specific traffic matrix as additional input. Such approaches thus separate the problem of the network deactivation and the actual routing into separately addressable problems, but of course are only feasible in practice if the underlying assumptions (a) and (b) are (at least mostly) valid.

\subparagraph{Our Contribution.}
In \Cref{sec:TO-algs} we discuss the two known traffic-oblivious concepts \probf{MCPS} and \probf{TOCA}, both of which still require to solve NP-hard problems: while there are established practical algorithms for the latter, we present the first algorithmic approach to compute exact solutions for the former.

Up to now it was impossible to truly estimate the absolute quality performance of either of the described solution methods, as there is no known way to directly solve the NP-hard true \prob{MSPND} or \prob{MSRND} problems. In \Cref{sec:mspnd}, we present the first exact approach to solve the former. It is based on a novel path-based integer linear program (ILP) formulation that requires a column generation scheme to solve. We further present strengthening inequalities that allow us to find provably optimal solutions for small to medium sized instances. Yet we can also conclude that already \prob{MSPND}, without the added complexity of \prob{MSRND}'s segment routing, is very challenging in practice.

Our algorithm allows us to conduct the first practical study on established network benchmark sets to gauge how much potential is wasted by following the \probl{F-} approach or either of the traffic-oblivious approaches, see \Cref{sec:experiments}.
The %
trivial \prob{F-}\prob{MSPND} approach performs exceptionally well in practice; in the traffic-oblivious scenario, simple \prob{TOCA}-algorithms perform best.

\section{Traffic-Oblivious Algorithms}\label{sec:TO-algs}

\subsection{Minimum Capacity-Preserving Subgraphs (\prob{MCPS})}\label{sec:mcps}
\newcommand{\cutout}{\ensuremath{\delta^-}}

Given a (directed) graph $G$ with arc capacities $\ecap\colon A\to \Npos$,
the \probf{MCPS} problem~\cite{CPSJournal} asks for an arc-minimum subgraph $H\subseteq G$ where
the capacities~$\maxflow{G}(s,t)$ and $\maxflow{H}(s,t)$ of minimum $s$-$t$-cuts (w.r.t.\ $\ecap$) in $G$ and $H$, respectively, differ by a factor of at most $\stretcha\in(0,1)$, i.e., $\maxflow{H}(s,t)\geq\stretcha\cdot\maxflow{G}(s,t)$ for all $(s,t)\in V^2$.
As the value of a minimum $s$-$t$-cut is equal to the value of a maximum
$s$-$t$-flow~\cite{DantzigFulkerson}, we have that, for each $(u,v)$ \emph{individually}, we may send at least $\stretcha$ as much flow in $H$ as in $G$.
In~\cite{CPSJournal}, exact algorithms to compute such subgraphs for specific
polynomial-time solvable instance families are established. For general instances, the possibility to adapt an ILP given by
Dahl~\cite{DBLP:journals/telsys/Dahl93,DBLP:journals/dam/Dahl93} for \probl{DSND} is briefly mentioned.

We adapt the scenario to our setting, taking multiple connections per link and possibly full-duplex links into account, and yield the following ILP that allows us to devise the first algorithm solving
\prob{MCPS} for general directed multigraphs. We call this algorithm \algcps{} in the following.
For a vertex subset $W\subset V$, let~$\cutout(W)\coloneqq \{uv\in \arcs : u\in W$, $v\not\in W\}$ be all arcs leaving $W$. Observe that $\maxflow{G}(s,t)$ is the size of a minimum-$s$-$t$-cut in $G$, where each arc $\aarc\in\arcs$ has individual capacity $\paredges{\aarc}\cdot\ccap(\aarc)$.
\begin{subequations}
\begin{align}
    \min &\sum_{\aarc \in \arcs}x_\aarc & \notag\\
    \sum_{\aarc \in \cutout(W)} \ccap(\aarc) \cdot x_\aarc &\geq \stretcha\cdot \maxflow{G}(s,t) & \forall (s,t)
    \in V^2, \forall W \subset V \colon s \in W, t \notin W\label{eq:ilp_mcps_cut}\\
    x_\aarc & \in \{0,1,...,\paredges{\aarc}\} &\forall \aarc \in \arcs \label{eq:ilp_mcps_integer_arc}
\end{align}
\end{subequations}

The integer $x$-variables, whose sum we minimize, directly map to our solution description $\numconns$.
In case of full-duplex links, we simply set $x_{uv} = x_{vu}$ for each link $\{u,v\}$.
The constraints~\eqref{eq:ilp_mcps_cut} for each
$(s,t)\in V^2$ ensure that all $s$-$t$-cuts in the solution (and thus in particular also the minimum ones) are sufficiently large.

Since the number of constraints~\eqref{eq:ilp_mcps_cut} is exponential in the input size, we need to facilitate a branch-and-cut scheme to speed up the ILP solving process. As per usual,
we start with an empty set of \eqref{eq:ilp_mcps_cut}~constraints and, whenever an LP relaxation in the
branch-and-bound tree is solved, add constraints that are violated by the
current model. Finding violated constraints is done via the following separation routine:
For each $\aarc\in\arcs$, let $\effcap(\aarc)
\be \ccap(\aarc)\cdot \hat x_\aarc$ be the \emph{effective capacity}
induced by the current solution~$\hat x$.
Then, for each terminal pair~$(s,t)\in V^2$, we compute a maximum
$s$-$t$-flow in $(G,\effcap)$ using Goldberg and Tarjan's preflow
algorithm~\cite{DBLP:journals/jacm/GoldbergT88}, terminating as soon as a flow
of value at least~$\stretcha\cdot \maxflow{G}(s,t)$ is found.
If no such flow exists, we compute two violated cuts and add their
respective constraints to the ILP model: the \emph{front cut} close to~$s$ and
the \emph{back cut} close to~$t$.
To compute the front cut, we traverse arcs in the residual network (i.e.,
forward arcs with positive residual capacity and backward arcs with positive
flow) via \ac{DFS} starting at~$s$ and mark all visited vertices.
Arcs of the front cut are then exactly those arcs~$uv$ where~$u$ was marked
and~$v$ can be reached via a backwards \ac{DFS} starting at~$t$ in the original (not the residual) graph using only unmarked vertices. The back cut is identified analogously with the roles of $s$ and $t$
swapped and running \acp{DFS} in the respective opposite directions. This process guarantees that if we do not find any constraint to add, then indeed all constraints are satisfied even if not actively generated into the model.

\newcommand{\remconns}{\ensuremath{\psi}} %
During separation, among cuts with the same cut value, we prefer those of lower
cardinality as this gives us a slight speed advantage. We achieve this by adding a miniscule positive value $\varepsilon< \nicefrac{1}{|\arcs|}$ to each arc capacity. This choice of $\varepsilon$ guarantees that if one cut has a smaller unperturbed capacity than another, the same still holds after perturbation.
We also include another optimization: Before starting the branch-and-cut
procedure, we precompute for each arc $\aarc=st\in\arcs$, whether the removal of some $\remconns_\aarc$-many connections of~$\aarc$ would lower
the maximum $s$-$t$-flow in $G$ to a value below $\stretcha\cdot
\maxflow{G}(s,t)$. If so, we can add the constraints $x_\aarc \geq \remconns_\aarc+1$.
Vertex pairs $(s,t) \in V^2$ whose capacity requirement is already reached by these lower bounds never have to be considered during
the separation routine.

\subsection{Traffic-Oblivious Cable Activation (TOCA)}\label{sec:toca}

In the \probf{TOCA} problem~\cite{DBLP:journals/corr/abs-2601-13087,DBLP:conf/isaac/ChimaniI25}, we ask for a subnetwork with minimum
number of active connections (again encoded by a vector $\numconns\in\mathbb{N}^A$) and the following property:
for every possible traffic matrix~$\demand$ that is MCF-routable in the full network $(G,\ccap,\len,\paredges{})$, the
scaled-down traffic matrix~$\mcfalpha\cdot\demand$ must be MCF-routable in
$(G,\ccap,\len,\numconns)$. Solving the underlying ILP turns out to be too time-consuming in practice~\cite{DBLP:journals/corr/abs-2601-13087}, so
we implement two algorithm variants proposed in \cite{DBLP:journals/corr/abs-2601-13087} based on LP rounding, which are supposed to yield close-to-optimum solutions:
First, a $\max(\frac{1}{\mcfalpha\cdot\minpar}, 2)$-approximation, where \(\minpar\coloneqq\min_{a\in\arcs} \paredges{a}\) is the minimum number of connections between any two
connected routers; we call this algorithm \algmcf{}.
It consists of computing a specific MCF LP where
the arc utilization is minimized, and rounding up the resulting fractional
variable values $x_\aarc$ for each arc~$\aarc\in\arcs$ to obtain the respective
number~$\numconns_\aarc\be\ceil{x_\aarc}$ of active connections.
Second, we also implemented a heuristic recommend
in~\cite{DBLP:journals/corr/abs-2601-13087} to further improve solutions
produced by \algmcf{}; we call it \algmcfpp{}:
After solving the MCF LP, set the lower (upper) bound of each
$x$-variable to the floor (ceiling, resp.) of its current value.
Then, as long as the solution is non-integral,
choose an arc~$\aarc$ with the smallest difference of its $x$-variable to the
next higher integer value, fix the $x$-variable to that value, and recompute the
LP.
As this heuristic computes up to $|\arcs|$ many LPs instead of only one, it
is a lot slower than \algmcf{}. Yet, it retains the approximation guarantee and produces feasible
solutions with significantly fewer connections.

\section{Exact \probl{MSPND}}\label{sec:mspnd}

We now devise the first ILP-based approach to solve the full \prob{MSPND} problem exactly. We make two assumptions:
First, the traffic demands are SPR-routable in the full network, without any deactivations.
As we consider the traffic demands to correspond to some low-traffic phase for which we want to deactivate superfluous connections, the fully activated network can safely be assumed to be capable of handling higher traffic during peak phases.
Second, we assume that shortest paths are unique, i.e., for any vertex pair~$(s,t)\in V^2$, the arc length function~$\len$ induces
a total order~$\pless$ on the $s$-$t$-paths. We can ensure this by enforcing an
arbitrary tie-breaking between $s$-$t$-paths of the same length.
We may start with considering the following ILP for \prob{MSPND}. For any $(s,t)\in K$, let~$\stpaths$ denote the set of all $s$-$t$-paths in~$G$.
\newcommand{\pv}{z} %
\newcommand{\iv}{y} %

\begin{subequations}
\begin{align}
    \min &\sum_{\aarc \in \arcs}x_\aarc & \notag
    \\
    \sum_{\ppath \in \stpaths} \pv_\ppath &\geq 1 & \forall (s,t) \in \termpairs \label{eq:ilp_mspnd_connectivity}\\
    \sum_{\substack{\ppath \in \stpaths \colon\\ \aarc \in \ppath}} \pv_\ppath &\leq \iv_\aarc  &\forall (s,t) \in \termpairs, \aarc\in \arcs \label{eq:ilp_mspnd_edgebuying}\\
    \sum_{(s,t) \in \termpairs}\sum_{\substack{\ppath \in \stpaths \colon\\ \aarc\in\ppath}}
    \demand(s,t) \cdot \pv_\ppath &\leq \ccap(\aarc) \cdot x_\aarc   &\forall \aarc\in\arcs \label{eq:ilp_mspnd_cap}\\
    |\ppath| - \sum_{\aarc \in \ppath} \iv_\aarc &\geq \sum_{\substack{\ppath' \in \stpaths \colon\\\ppath \pless \ppath'}} \pv_{\ppath'}
        &\forall (s,t) \in \termpairs, \ppath \in \stpaths \label{eq:ilp_mspnd_shortest}\\
    \iv_\aarc &\leq x_\aarc &\forall\aarc\in \arcs \label{eq:ilp_mspnd_conn_down}\\
    x_\aarc  &\leq \paredges{\aarc}\cdot \iv_\aarc &\forall\aarc\in \arcs \label{eq:ilp_mspnd_conn_up}\\    %
    x_\aarc & \in \{0,1,\dots,\paredges{\aarc}\} &\forall \aarc \in \arcs\label{eq:ilp_mspnd_integer_conn}\\
    \iv_\aarc & \in \{0,1\} &\forall \aarc \in \arcs \label{eq:ilp_mspnd_integer_arc}\\
    \pv_\ppath & \in \{0,1\} &\forall (s,t) \in \termpairs, \ppath \in \stpaths \label{eq:ilp_mspnd_integer_path}
\end{align}
\end{subequations}
The integer variables $x_\aarc$ encode how many connections per arc are active. They directly map to our solution vector $\numconns$ and are thus minimized.
The binary variables $\iv_\aarc$ indicate whether any connection along $\aarc$ is active, in which case we may say the \emph{arc} is active. This property is enforced via \eqref{eq:ilp_mspnd_conn_down} and \eqref{eq:ilp_mspnd_conn_up}.

For every $(s,t)\in\termpairs$, we have a binary $\pv_\ppath$ variable for each $\ppath\in\stpaths$, which indicates whether~$\ppath$ is active, i.e., whether it is the
shortest $s$-$t$-path in the solution subgraph.
Constraint~\eqref{eq:ilp_mspnd_connectivity} establishes that each terminal pair
is connected by an active path in the solution;
constraint~\eqref{eq:ilp_mspnd_edgebuying} guarantees that all arcs on an active path are
active as well;
constraint~\eqref{eq:ilp_mspnd_cap} ensures that the total traffic
routed over any arc does not exceed its capacity.
Correct shortest path routing is achieved via constraint~\eqref{eq:ilp_mspnd_shortest}: it ensures that, when all arcs on an $s$-$t$-path~$\ppath$ are active, no $s$-$t$-path longer than $\ppath$ can be active.

Above, we considered general simplex-connections;
for full-duplex links,
we simply enforce $x_{st} = x_{ts}$ and $\iv_{st} = \iv_{ts}$ for all $st\in\arcs$.
We could also enforce $\pv_{\ppath} = \pv_{\ppath'}$  for all paths~$\ppath$ and their
respective reverse paths~$\ppath'$, however, a pilot study showed that
this only slows the algorithm down when using the strengthening
inequalities of \Cref{sec:strength}.

In any case, we obtain the LP relaxation by replacing
constraints~\eqref{eq:ilp_mspnd_integer_conn}--\eqref{eq:ilp_mspnd_integer_path} by:
\begin{align}
    0 &\leq x_\aarc \leq \paredges{\aarc} &\forall \aarc \in \arcs \tag{\ref{eq:ilp_mspnd_integer_conn}'}\label{eq:ilp_mspnd_frac_conn}\\
    0 &\leq \iv_\aarc \leq 1 &\forall \aarc \in \arcs \tag{\ref{eq:ilp_mspnd_integer_arc}'}\label{eq:ilp_mspnd_frac_arc}\\
    0 &\leq \pv_\ppath  &\forall (s,t) \in \termpairs, \ppath \in \stpaths \tag{\ref{eq:ilp_mspnd_integer_path}'}\label{eq:ilp_mspnd_frac_path}
\end{align}
The constraint~\eqref{eq:ilp_mspnd_edgebuying} already implies $\pv_\ppath \leq 1$
for all paths~$\ppath$.

Our ILP contains a variable %
for every path in the
input digraph and thus the total number of variables is exponential in $|A|$. This prohibits to straight-forwardly solve the ILP with standard means.
We follow the \emph{branch-and-price} paradigm~\cite{DBLP:books/daglib/0023873}: Starting with an initial \emph{model} that only contains a subset of the variables (and constraints), we repeatedly solve the LP relaxation and add variables not yet in the model whose addition may allow to improve the current solution value (a process called \emph{column generation} or \emph{pricing}). When the latter fails, we resort to \emph{branching}, i.e., we pick any variable with fractional value $f$ and generate two subproblems which enforce $\leq\lfloor f\rfloor$ or $\geq\lceil f\rceil$ on the domain of the variable, respectively.

We observe that in our ILP we never have to enforce integrality on the $\pv$-variables: whenever we have an optimal solution with integral $x$- and $\iv$-variables, it is never beneficial to choose a non-binary value for any $\pv$-variable.
Thus, our scheme never branches on $\pv$-variables.

\subsection{Column Generation}

Our initial model contains all $x$- and $\iv$-variables.
For every terminal pair~$(s,t) \in\termpairs$, let $\stpaths'\subseteq \stpaths$ denote the set of $s$-$t$-paths whose $\pv$-variables are in the current model. Initially, $\stpaths'$ contains the five shortest $s$-$t$-paths in~$G$.
Whenever a new LP solution is found, we start our \emph{pricing routine} to identify and add beneficial $\pv$-variables. When adding a variable to the model, observe that this automatically also adds the corresponding constraint~\eqref{eq:ilp_mspnd_shortest} and updates the other constraints according to the now larger subsets $\stpaths'$.

We consider the dual program to our LP.
As per usual, a variable is (potentially) beneficial if and only if its dual constraint is violated by the current dual solution. If no such variable exists, we have found an LP solution that is optimal w.r.t.\ the full LP, not only the current model~\cite{DBLP:books/daglib/0023873}.
\newcommand{\dualcon}{\ensuremath{\alpha}}%
\newcommand{\dualcap}{\ensuremath{\gamma}}%
\newcommand{\dualeb}{\ensuremath{\beta}}%
\newcommand{\dualsht}{\ensuremath{\delta}}%
For brevity, and since it can canonically be obtained from the above LP, we refrain from presenting the full dual program here. We are only interested in the dual constraints corresponding to our (primal) $\pv$-variables.
Let $\dualcon$, $\dualeb$, $\dualcap$, and $\dualsht$ be the
dual variables corresponding to the constraints \eqref{eq:ilp_mspnd_connectivity}, \eqref{eq:ilp_mspnd_edgebuying}, \eqref{eq:ilp_mspnd_cap}, and \eqref{eq:ilp_mspnd_shortest}, respectively. Recall that there are no explicit upper bound constraints $\pv_\ppath\leq 1$, so the dual constraint is:
\begin{equation}
    \dualcon_{s,t}
    - \sum_{\aarc \in \ppath} \dualeb_{s,t,\aarc}
    - \demand(s,t) \cdot \sum_{\aarc \in \ppath} \dualcap_{\aarc}
    - \sum_{\ppath' \in \stpaths\colon \ppath' \pless \ppath}
    \dualsht_{s,t,{\ppath'}}
    \leq 0 \quad
    \forall (s,t) \in \termpairs, \ppath\in\stpaths. \label{eq:dual_mspnd_path}
\end{equation}

Consequently, for each terminal pair~$(s,t)\in\termpairs$, we search for a path~$\ppath$ that
violates this dual constraint.
At the same time, we aim to find paths of the optimal solution early
on to speed up the solving process, and such paths tend to be short w.r.t.\ the
arc length~$\len$.
Hence, our \emph{pricing problem} consists of finding a shortest path (w.r.t.~$\len$)
among those paths whose total dual cost is below the limit~$\dualcon_{s,t}$.
This kind of pricing problem, the \emph{\ac{CSP}} problem,
is commonly found in column generation algorithms for network design problems
such as spanners~\cite{DBLP:conf/esa/SigurdZ04}.
However, there is a crucial difference in our pricing problem that complicates
the process:
Since the variables $\dualsht_{s,t,{\ppath'}}$ are not only dependent on the
terminal pair~$(s,t)$ but also on other paths~$\ppath'$ we must define our
\emph{dual cost function}
$\cost(\aarc) \be \sum_{\aarc \in \ppath} \dualeb_{s,t,\aarc} + \demand(s,t) \cdot \sum_{\aarc \in \ppath} \dualcap_{\aarc}$, for all~$\aarc\in\arcs$, without them.
Since all dual values are non-negative, the set $\stpaths^{\leq\cost}$  of $\cost$-feasible $s$-$t$-path is thus a strict superset of those paths that have a truly violated dual constraint.
It also contains paths that are not violated and thus may already be in our current model. In particular, the shortest path (w.r.t.\ $\len$) among $\stpaths^{\leq\cost}$ may indeed be of that type. In our scheme, we need to account for this to ensure that we identify a path of $\stpaths^{\leq\cost}\setminus\stpaths'$ if any exist.

\newcommand{\queue}{\ensuremath{Q}}
\newcommand{\heap}{\ensuremath{H}}
\newcommand{\distlen}{\ensuremath{\overline{\len}}}
\newcommand{\distcost}{\ensuremath{\overline{\cost}}}
We solve the pricing problem for each terminal pair~$(s,t)\in\termpairs$ using a
label setting algorithm:
First, we may abort the algorithm immediately if already the previous attempt (w.r.t.\ $(s,t)$) was unsuccessful and we can further infer from the
total reductions in arc costs and cost limit that a new path adhering to the
new cost limit cannot exist.
Otherwise, for all vertices~$v\in V$, we precompute
the cost~$\distcost(v,t)$ of a $\cost$-minimum $v$-$t$-path (disregarding $\len$)
and the length~$\distlen(s,v)$ of a shortest (w.r.t.\ $\len$)  $s$-$v$-path (disregarding $\cost$).
Then, during the actual path computation, we manage both a global priority
queue~$\queue$ with vertices and a heap~$\heap_v$ with labels at every
vertex~$v$.
Each such label~$\ell$ corresponds to a path~$\ppath_\ell$ from~$s$ to
vertex~$v(\ell)$ with
 length~$\distlen(\ell) \be \sum_{\aarc \in \ppath_\ell} \len(\aarc)$
and dual cost~$\distcost(\ell) \be \sum_{\aarc \in \ppath_\ell} \cost(\aarc)$.
When queried for its top element, a heap~$\heap_v$ returns the label~$\ell$ with
the lexicographical minimum according to~$(\distlen(\ell) - \distlen(s,v(\ell)),\ \distlen(s,v(\ell)),\ \distcost(\ell))$,
and~$\queue$ returns the vertex with the lexicographical minimum label among all heaps.
We thus first process one label per vertex~$v$ in the order of~$\distlen(s,v)$, i.e., in the order of Dijkstra's algorithm~\cite{DBLP:journals/nm/Dijkstra59}, before moving on to other labels. This allows us to find the shortest (w.r.t.\ $\len$) path
violating \eqref{eq:dual_mspnd_path}
first.

While $\queue$ is non-empty, we pop the top vertex~$v(\ell)$ from~$\queue$
and the top label~$\ell$ from~$\heap_{v(\ell)}$, reinserting~$v(\ell)$
into~$\queue$ only if $\heap_{v(\ell)}$ contains more labels to process.
If
$v(\ell) = t$, we check whether
path~$\ppath_\ell\in\stpaths'$ already, and accordingly discard $\ell$ or
return~$\ppath_\ell$.
Otherwise, we scan all of $v(\ell)$'s neighbors~$w$ where $w$ is not yet
in~$\ppath_\ell$ (to avoid cycles) and~$w$ satisfies~$\distcost(\ell) +
\cost(v(\ell)w) + \distcost(w,t) < \dualcon_{s,t}$.
Then, we create a new label $\ell'$ where $\ppath_{\ell'}$ is obtained by
appending~$v(\ell)w$ to~$\ppath_\ell$, and insert~$\ell'$ into~$\heap_{w}$ (and $w$ into $\queue$ if it is not yet contained).

In contrast to other \ac{CSP} problems, we cannot simply discard a label~$\ell'$ if it is \emph{dominated} (i.e., another label in $\heap_{v(\ell')}$ has smaller cost and length), as
dominating labels may only expand to $s$-$t$-paths that already are
in~$\stpaths'$.
However, to speed up the algorithm, we run the pricing routine for $(s,t)$ once
discarding dominated labels (and thus avoiding the need to explicitly check for
cycles), and only if no new paths are found this way, we run it again without the domination check.

\subsection{Strengthening Inequalities}\label{sec:strength}

A (feasible) constraint \emph{polyhedrally strengthens}~\cite[Def.~13.2]{DBLP:books/daglib/0023873} a minimizing ILP formulation if there are instances such that the value of the corresponding LP relaxation increases when adding the constraint to the model. In other words, while not necessary from the point of integer feasible solutions, adding the constraint to the model can yield stronger bounds in the branch-and-bound framework, and thus potentially lead to faster running times.

Let $\subpath{\ppath}{u,v}$ denote the subpath of a path~$\ppath$ that starts at vertex~$u$ and ends at vertex~$v$.
While not necessary to model
\prob{MSPND}, these inequalities strengthen the ILP formulation by establishing the
well-known property that subpaths of shortest paths are also shortest~paths:
\begin{align}
        \pv_{\subpath{\ppath}{s,v}} &\geq \pv_{\ppath}
&        \pv_{\subpath{\ppath}{v,t}} &\geq \pv_{\ppath}
             &\forall (s,t) \in \termpairs, \ppath \in \stpaths, v \in V(\ppath)
 \label{eq:ilp_mspnd_subpath}
\end{align}

\newcommand{\sol}{\ensuremath{S}}
\newcommand{\satarc}{\ensuremath{\bar{\aarc}}}
\newcommand{\Athin}{\ensuremath{\arcs_{\text{thin}}}}
\newcommand{\Athick}{\ensuremath{\arcs_{\text{thick}}}}
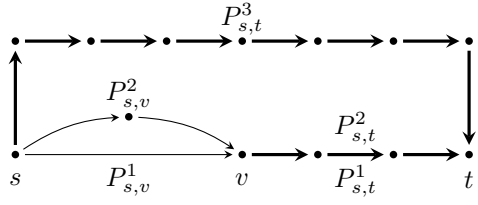
\begin{figure}[tbp]
    \centering
    \begin{tikzpicture}[scale=1]
            \begin{scope}[every node/.style={dot}]
                \node[label={[label distance=1mm]270:$s$}] (a) at (0,0) {};
                \node (b) at (1.5,0.5) {};
                \node[label={[label distance=1mm]270:$v$}] (c) at (3,0) {};
                \node (d) at (4,0) {};
                \node (e) at (5,0) {};
                \node[label={[label distance=1mm]270:$t$}] (f) at (6,0) {};
                \node (a0) at (0,1.5) {};
                \node (a1) at (1,1.5) {};
                \node (b1) at (2,1.5) {};
                \node (c1) at (3,1.5) {};
                \node (d1) at (4,1.5) {};
                \node (e1) at (5,1.5) {};
                \node (e0) at (6,1.5) {};
            \end{scope}
            \node[draw=none,above=-0.1 of b] {$\ppath^2_{s,v}$};
            \node[draw=none,above=-0.1 of c1] {$\ppath^3_{s,t}$};

            \begin{scope}[every edge/.style={edge, thin}]
                \path (a) edge[bend left=15] (b);
                \path (b) edge[bend left=15] (c);
                \path (a) edge node [below] {$\ppath^1_{s,v}$}  (c);
            \end{scope}

            \begin{scope}[every edge/.style={edge}]
                \path (c) edge (d);
                \path (d) edge node [below] {$\ppath^1_{s,t}$} node [above] {$\ppath^2_{s,t}$} (e);
                \path (a) edge (a0);
                \path (a0) edge (a1);
                \path (a1) edge (b1);
                \path (b1) edge (c1);
                \path (c1) edge (d1);
                \path (d1) edge (e1);
                \path (e1) edge (e0);
                \path (e0) edge (f);
                \path (e) edge (f);
            \end{scope}
    \end{tikzpicture}
    \caption{\prob{MSPND} instance considered in the proof of \Cref{th:strengthening_inequality}.
        Arcs $\Athin$ ($\Athick$) are drawn as thin (thick) lines and have capacity 1 (2, respectively).}
    \label{fig:strengthening_inequality}
\end{figure}

\begin{theorem}\label{th:strengthening_inequality}
    Constraint~\eqref{eq:ilp_mspnd_subpath} is polyhedrally strengthening, even when all arc lengths are~1.
\end{theorem}
\begin{proof}
    Consider the digraph $G=(V,\arcs)$ and its vertices $s,v,t\in V$ as shown in
    \Cref{fig:strengthening_inequality}.
    For $i \in \{1,2\}$, let $\ppath^i_{s,v}$ be the unique $s$-$v$-path
    containing exactly $i$~arcs and $\ppath^i_{s,t}$ the unique $s$-$t$-path
    containing $\ppath^i_{s,v}$; further, let be $\ppath^3_{s,t}$ the last
    remaining $s$-$t$-path.
    Lastly, let $\Athin \be \Union_{i \in \{1,2\}} \ppath^i_{s,v}$ be the three arcs on $s$-$v$-paths and $\Athick \be \arcs \setminus \Athin$.
    Consider the \prob{MSPND} instance on network~$(G,\ccap,\len,\paredges{})$ and traffic~$\demand$ with
    $\len(\aarc) = \paredges{\aarc} = 1, \forall\aarc \in \arcs$, and
    \begin{align*}
        \ccap(\aarc) = \ecap(\aarc) &= \begin{cases}
            1 & \text{if } \aarc \in \Athin,\\
            2 & \text{if } \aarc \in \Athick;
        \end{cases}&
        \demand(u,w) &= \begin{cases}
            2 & \text{if } (u,w) = (s,t),\\
            1 & \text{if } (u,w) = (s,v),\\
            0 & \text{otherwise.}
        \end{cases}
    \end{align*}
    The instance allows a feasible integer solution (of value $9$) by picking precisely the single arc on $\ppath^1_{s,v}$ and all arcs on $\ppath^3_{s,t}$.
    The following fractional solution $(\bar x,\bar \iv,\bar \pv)$ has only objective value $8.5$ and satisfies all constraints of the original LP relaxation:
    \begin{align*}
        \bar x_\aarc = \bar \iv_\aarc &= \begin{cases}
            1 & \text{if } \aarc \in \Athin,\\
            \nicefrac{1}{2} & \text{if } \aarc \in \Athick;
        \end{cases}&
        \bar \pv_\ppath &= \begin{cases}
            1 & \text{if } \ppath = \ppath^1_{s,v},\\
            \nicefrac{1}{2} & \text{if } \ppath \in \{\ppath^2_{s,t},\ppath^3_{s,t}\},\\
            0 & \text{otherwise.}
        \end{cases}
    \end{align*}
    Since
    $\bar \pv_{\ppath^2_{s,v}} < \bar \pv_{\ppath^2_{s,t}}$,
    adding constraint~\eqref{eq:ilp_mspnd_subpath} would
    cut off this solution. %

    It remains to argue that the LP including~\eqref{eq:ilp_mspnd_subpath} enforces an objective value strictly larger than $8.5$. Since $\paredges{}=1$, we have $x_\aarc=\iv_\aarc$ for each $\aarc\in A$ by \eqref{eq:ilp_mspnd_conn_down} and \eqref{eq:ilp_mspnd_conn_up}.
    A key ingredient is that $\pv_\ppath$, $0\leq \pv_\ppath\leq1$, represents a traffic of volume $\pv_\ppath$ if $\ppath$ is an $s$-$v$-path (since $\demand(s,v)=1$), but a traffic of volume $2\pv_\ppath$ if $\ppath$ is an $s$-$t$-path (since $\demand(s,t)=2$).
    Let $\aarc_i\in \ppath^i_{s,v}$, $i\in\{1,2\}$.
    By the capacity constraint \eqref{eq:ilp_mspnd_cap} for $\aarc_i$ we have
    \begin{equation}
        x_{\aarc_i}\geq \pv_{\ppath^i_{s,v}}+2\pv_{\ppath^i_{s,t}}.\label{eq:X}
    \end{equation}
    By the shortest-path constraint~\eqref{eq:ilp_mspnd_shortest} for $\ppath^1_{s,v}$ we get $\pv_{\ppath^2_{s,v}}\leq 1-x_{\aarc_1}$ and together with~\eqref{eq:X} for $\aarc_1$ we have
$\pv_{\ppath^2_{s,v}}\leq 1-\pv_{\ppath^1_{s,v}}-2\pv_{\ppath^1_{s,t}}$.
    Since $\pv_{\ppath^1_{s,v}}+\pv_{\ppath^2_{s,v}}\geq 1$ this simplifies to $0\leq -2\pv_{\ppath^1_{s,t}}$ and consequently $\pv_{\ppath^1_{s,t}}=0$.
    By the strengthening~\eqref{eq:ilp_mspnd_subpath}, we have $\pv_{\ppath^2_{s,t}}\leq \pv_{\ppath^2_{s,v}}$.
    Let $\pv_{\ppath^2_{s,v}}-\pv_{\ppath^2_{s,t}}=\varepsilon\geq 0$.
    For any $\aarc_2\in\ppath^2_{s,v}$ we can use \eqref{eq:X} to thus deduce
    $x_{\aarc_2}\geq 3\pv_{\ppath^2_{s,t}} +\varepsilon$. For $\aarc_1\in\ppath^1_{s,v}$ we know
    $x_{\aarc_1}=\pv_{\ppath^1_{s,v}}\geq1-\pv_{\ppath^2_{s,v}} = 1-\pv_{\ppath^2_{s,t}}-\varepsilon$.
    Observe that for any arc $\aarc\in A$, we have $x_\aarc=\iv_\aarc$ at least large enough to satisfy constraint~\eqref{eq:ilp_mspnd_cap}, where exactly the arcs~$\Athick$
    have capacity $2$. So we can deduce the objective function value as
    \begin{align*}
        \sum_{\aarc\in A} x_\aarc &=
        \sum_{\aarc\in P^3_{s,t}} x_\aarc
         + \sum_{\aarc\in P^2_{s,t}\setminus P^2_{s,v}} x_\aarc
         + \sum_{\aarc\in P^2_{s,v}} x_\aarc
         + \sum_{\aarc\in P^1_{s,v}} x_\aarc\\
        & \geq |P^3_{s,t}|\cdot(\demand(s,t)\cdot \pv_{\ppath^3_{s,t}}/2)
        \ +\  |P^2_{s,t}\setminus P^2_{s,v}|\cdot(\demand(s,t)\cdot \pv_{\ppath^2_{s,t}}/2)\\
        &\hspace{1cm}+\ |P^2_{s,v}|\cdot (3\pv_{\ppath^2_{s,t}} +\varepsilon)
        \ +\  |P^1_{s,v}|\cdot (1-\pv_{\ppath^2_{s,t}}-\varepsilon)\\
        & = 8 (1-\pv_{\ppath^2_{s,t}}) + 3 \pv_{\ppath^2_{s,t}} + 2 (3\pv_{\ppath^2_{s,t}} +\varepsilon) + 1(1-\pv_{\ppath^2_{s,t}}-\varepsilon)
        = 9 +\varepsilon \geq 9.\qedhere
    \end{align*}

\end{proof}

We do not need to add constraints for~$|\subpath{\ppath}{s,v}| = 1$
and~$|\subpath{\ppath}{v,t}| = 1$ whenever
the arc length function~$\len$ is \emph{one-shortest}, i.e., if two vertices $s,t\in V$ are connected by a direct arc $st\in A$, there is no shorter path than $st$.

\begin{observation}
    For one-shortest arc lengths, the
    constraints~\eqref{eq:ilp_mspnd_subpath} for $|\subpath{\ppath}{s,v}| = 1$
    and $|\subpath{\ppath}{v,t}| = 1$
    are already implied by the original LP relaxation.
\end{observation}
\begin{proof}
    We only argue the case $|\subpath{\ppath}{s,v}| = 1$; $|\subpath{\ppath}{v,t}| = 1$ is analogous.
    Let $\aarc=sv$ be the first edge on some path $\ppath\in\stpaths$ with $|\ppath|>1$. Clearly, $\aarc$ itself is a path $\ppath'\in\pathset_{s,v}$, and so
\[\pv_\ppath
\overset{\eqref{eq:ilp_mspnd_edgebuying}}{\leq}
    \iv_\aarc
\overset{\eqref{eq:ilp_mspnd_shortest}}{\leq}
    1 - \sum_{{\bar{\ppath} \in \pathset_{s,v} \colon \ppath' \pless \bar{\ppath}}} \pv_{\bar{\ppath}}
\overset{\eqref{eq:ilp_mspnd_connectivity}}\leq
    1-\big(1-\pv_{\ppath'}-\sum_{{\hat{\ppath} \in \stpaths \colon \hat{\ppath} \pless \ppath'}} \pv_{\hat{\ppath}}\big)
    \leq \pv_{\ppath'},
\]
where the last sum evaluates to $0$ since there are no path $\hat{\ppath}\pless P'$ as $\len$ is one-shortest.
\end{proof}

\section{Experiments}\label{sec:experiments}

We consider our column generation scheme \algspnd{} for exact solutions of \prob{MSPND} (\Cref{sec:mspnd}) together with the traffic-oblivious algorithms \algcps{}, \algmcf{}, and \algmcfpp{} presented in \Cref{sec:TO-algs}. Furthermore, we consider  the trivial algorithm \algspndone{} that computes
an optimal solution for \prob{F-}\prob{MSPND}.
All algorithms are implemented in
\texttt{C++}, using functionality of the Open Graph Drawing
Framework (OGDF~\cite{DBLP:reference/crc/ChimaniGJKKM13}, \url{www.ogdf.net}) release \enquote{2025.10 Foxglove}.
The implementations will be made public as part of the next OGDF release but
are already available---alongside the complete dataset and detailed results discussed herein---at \url{tcs.uos.de/research/spnd}.
The code is compiled with GCC~12.2.0 using optimization level \texttt{-O3}.
Each computation is performed on a single physical processor of a
Xeon~Gold~6134~CPU (3.2~GHz) running Debian GNU/Linux 12 \enquote{Bookworm}
with Linux kernel version 6.1.0-18-amd64.
LP computations are performed using Gurobi~11.0.1; \algspnd{} uses SCIP~9.0.0 as a branch-and-price framework.
We enforce a time limit of 10~minutes for each run.

\subsection{Instances}
We use the three instance sets of the REPETITA framework for
\enquote{repeatable experiments on Traffic Engineering
algorithms}~\cite{DBLP:journals/corr/abs-1710-08665}, which contain network
topologies including arc capacities and lengths, as well as traffic matrices that are designed for
a \ac{MLU} of 90\% when routed via MCF in the original network topology.

\begin{itemize}
    \item \textbf{DEFO}~\cite{DBLP:conf/sigcomm/HartertVSBFTF15}: Nine network
        topologies
        with
        50--315 vertices and 276--1944 arcs, synthetically generated for
        benchmarking traffic engineering algorithms, and one
        traffic matrix each.
    \item \textbf{Rocketfuel}~\cite{DBLP:journals/ton/SpringMWA04}: Six network
        topologies with 79--315 vertices and 294--1944 arcs, collected using
        Internet mapping techniques and each one accompanied by five traffic matrices.
    \item \textbf{Topology Zoo (TZ)}~\cite{DBLP:journals/jsac/KnightNFBR11}: 260
        real-world topologies with 4--197 vertices and 8--490 arcs as reported by
        official company websites, each one accompanied by five different
        traffic matrices.
        Since the original data set does not contain arc lengths, we consider it once with unit arc lengths and once with arc lengths that are inversely
        proportional to the arc capacities, as these are the standard link metrics used in practice~\cite{DBLP:journals/corr/abs-1710-08665}.
\end{itemize}

For each combination of network topology and traffic matrix, we consider the
network topology once in the simplex and once in the full-duplex communication setting. %
We preprocess each instance by removing all parallel arcs,
enforcing the same arc length in both directions of each link for the full-duplex setting, and scaling down the traffic matrix such that SPR on the
instance produces an \ac{MLU} of exactly 1.

We consider two further parameters for each network topology:
the number of connections per link~$\paredges{}\in\{1,5\}$, and
the retention ratio~$\mcfalpha \in \{0.3,0.5,0.7\}$.
Given an SPR-feasible traffic matrix $\demand$ for the full network, we consider the \prob{MSPND} instance with scaled matrix $\mcfalpha \demand$ to simulate phases of low and medium traffic.
Recall that \algcps{}, \algmcf{}, and \algmcfpp{} are oblivious to the traffic matrix (and the arc lengths) but require~$\mcfalpha$ as an input parameter.

\subsection{Results}
We examine the implemented algorithms w.r.t.\ success rate, running time, number of active connections in the solution, and maximum link utilization.
Our main observations hold across all instance sets; also, the choice of arc lengths for the TZ instances has surprisingly no significant impact on the results. In the following, we often focus on the largest set~TZ.

\begin{figure}
    \centering
    \includegraphics[width=0.49\textwidth]{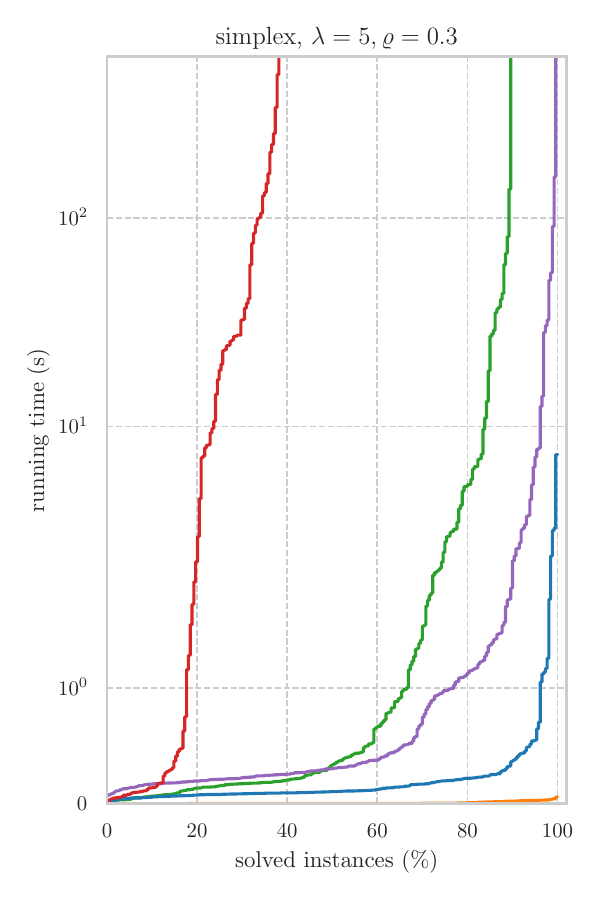}%
    \includegraphics[width=0.49\textwidth]{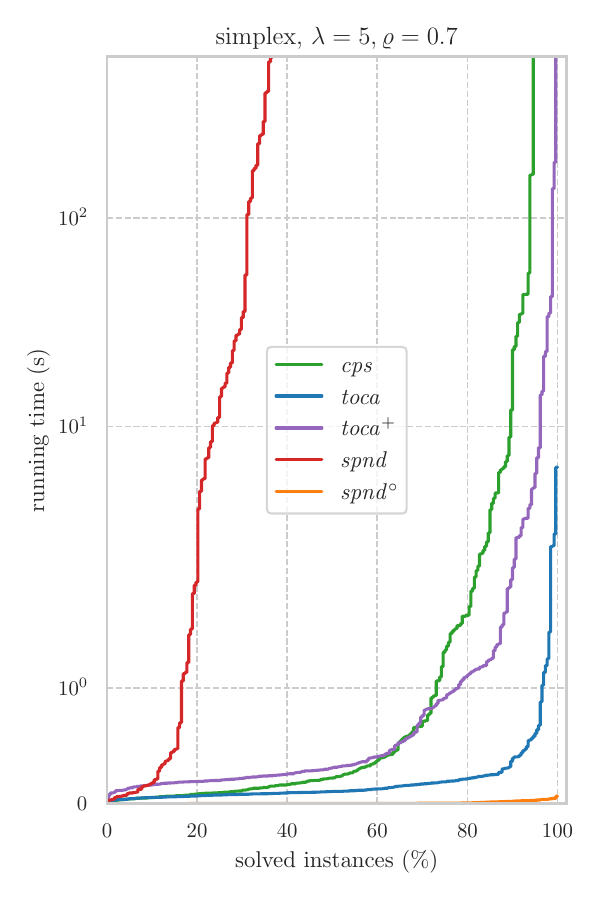}\\
    \includegraphics[width=0.49\textwidth]{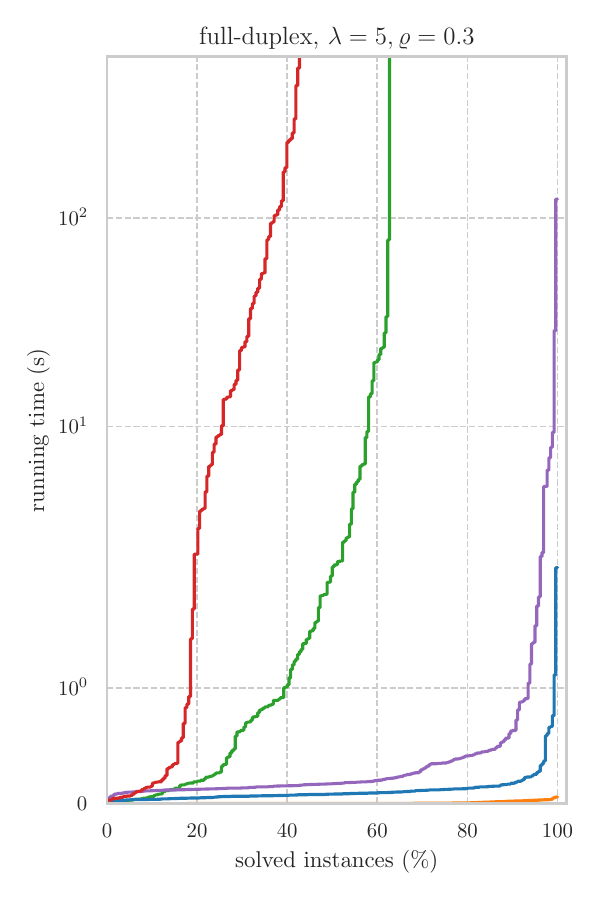}%
    \includegraphics[width=0.49\textwidth]{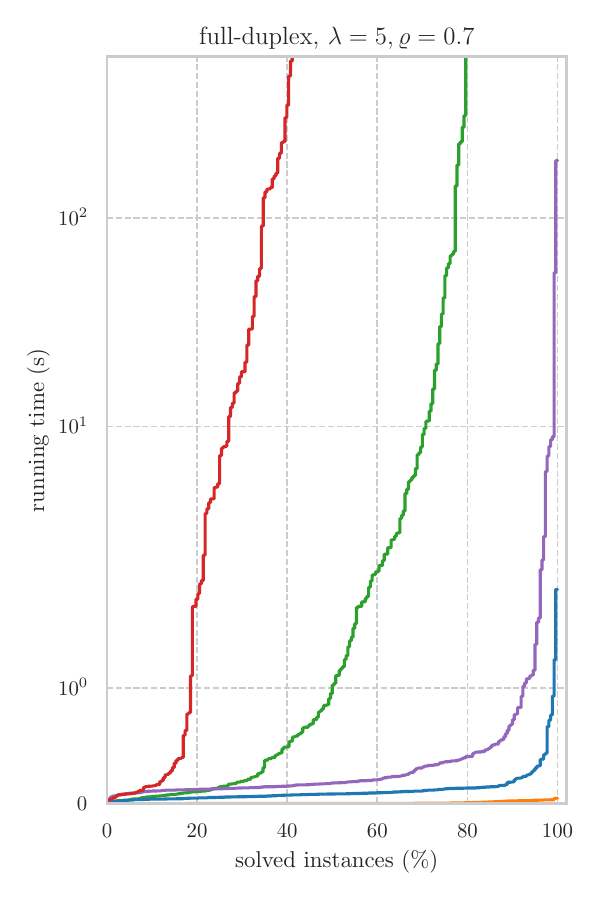}
    \caption{Performance plots for representative parameter settings.}
    \label{fig:performance}
\end{figure}

\begin{figure}
    \centering
    \includegraphics[width=0.49\textwidth]{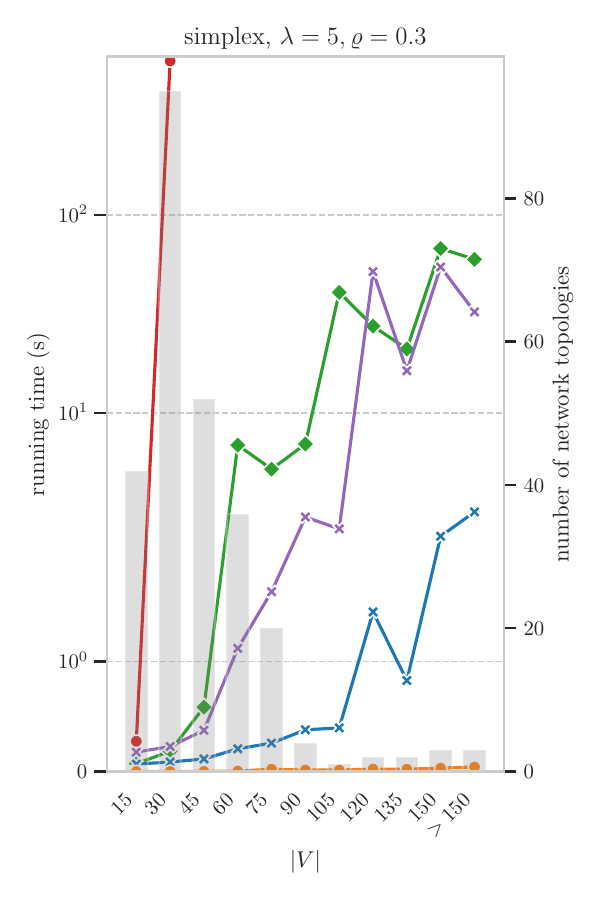}%
    \includegraphics[width=0.49\textwidth]{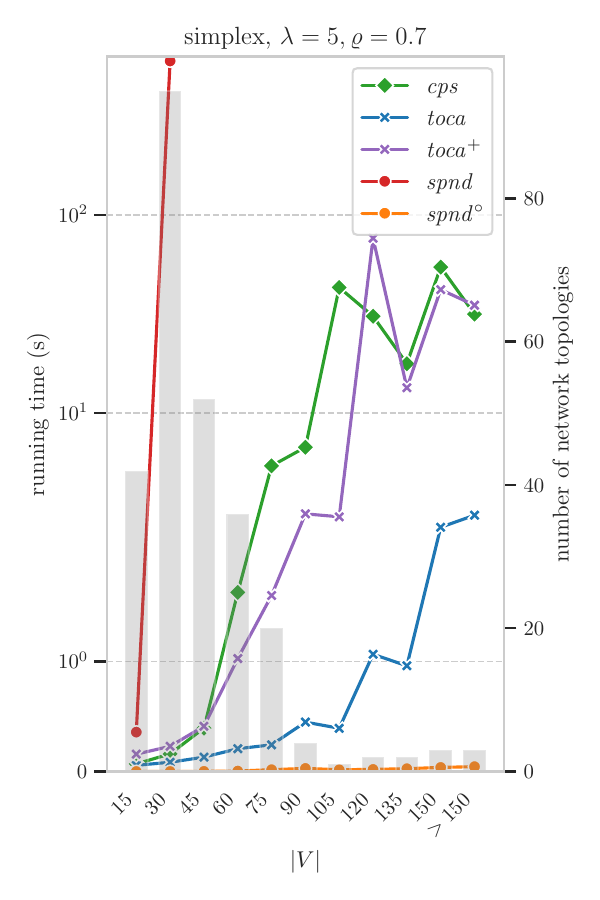}\\
    \includegraphics[width=0.49\textwidth]{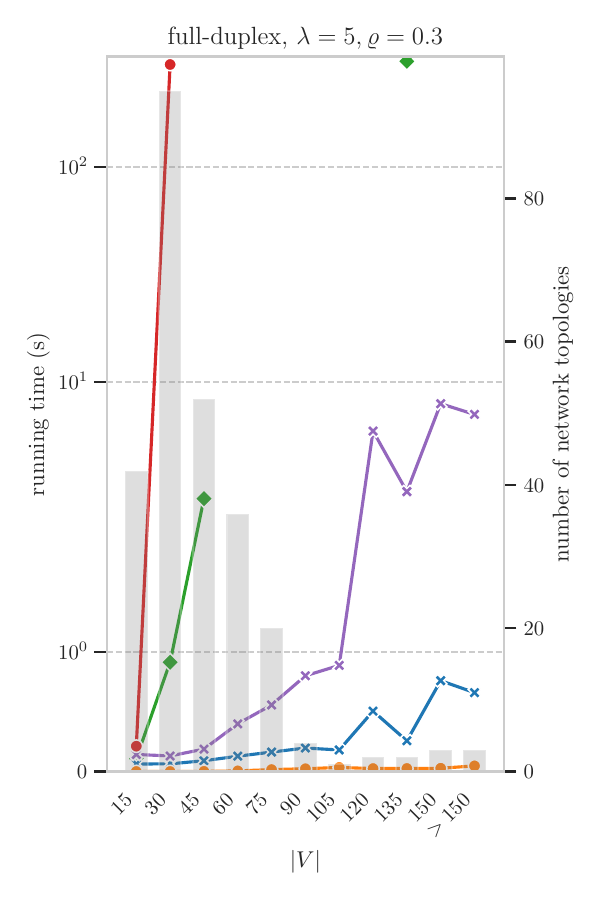}%
    \includegraphics[width=0.49\textwidth]{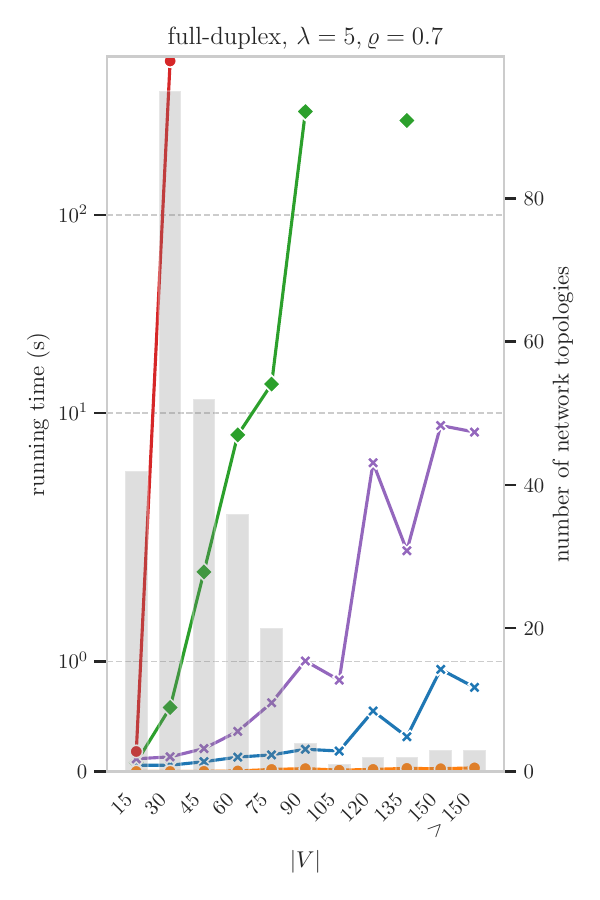}
    \caption{Median running times for representative parameter settings. Timeouts are not displayed.}
    \label{fig:running_time_plots}
\end{figure}

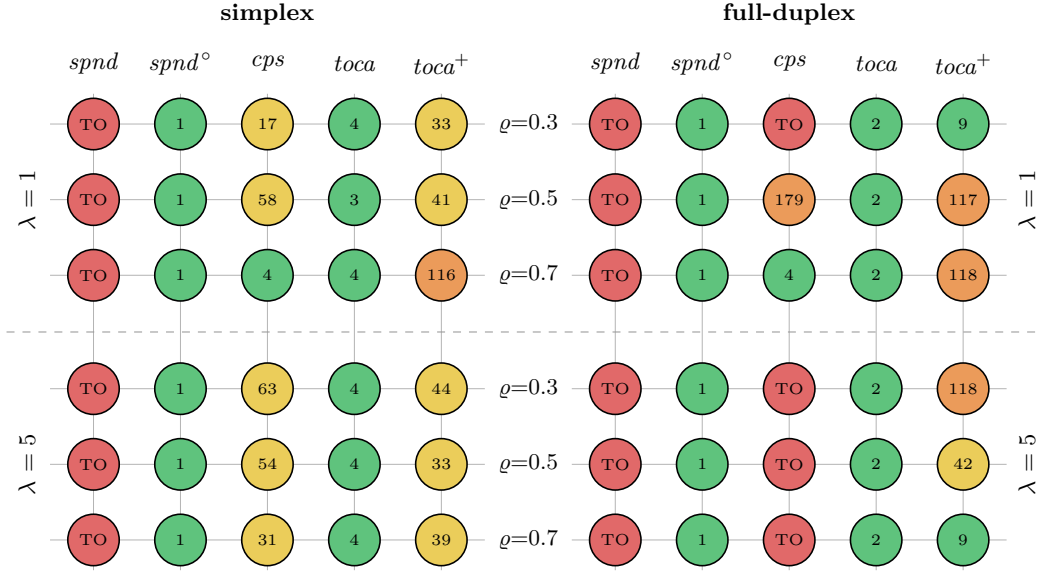
\begin{figure}
    \centering
    \definecolor{colTO}   {RGB}{225, 105, 105}
\definecolor{col1dig} {RGB}{ 95, 195, 125}
\definecolor{col2dig} {RGB}{235, 205,  90}
\definecolor{col3dig} {RGB}{235, 155,  90}
\definecolor{gridCol} {RGB}{180,180,180}

\newcommand{\pickcolor}[1]{%
	\def\theCC{colTO}%
	\if T#1\else
	\pgfmathparse{int(#1)}%
	\ifnum\pgfmathresult<10  \def\theCC{col1dig}\else
	\ifnum\pgfmathresult<100 \def\theCC{col2dig}\else
	\def\theCC{col3dig}%
	\fi\fi
	\fi
}

\newcommand{\dcircle}[4]{%
	\pickcolor{#3}%
	\filldraw[fill=\theCC, draw=black, line width=0.6pt]
	(#1,#2) circle (0.34cm);
	\node[font=\fontsize{6}{7}\selectfont, text=black]
	at (#1,#2) {#4};
}

\def\ySepA{3.75}
\def\yI{1.0}  \def\yII{2.0}  \def\yIII{3.0}
\def\yIV{4.5} \def\yV{5.5}   \def\yVI{6.5}

\def\xAlpha{6}
\def\xRshift{6}
\def\xEpL{0.25}
\def\xEpR{11.75}

\begin{tikzpicture}[x=1.15cm, y=1.0cm]
    \node[font=\small\bfseries] at (3.0, 0.45) {simplex};
    \node[font=\small\bfseries] at (9.0, 0.45) {full-duplex};

    \foreach \lbl/\cx in {
        \algspnd/1,
        \algspndone/2,
        \algcps/3,
        \algmcf/4,
        \algmcfpp/5}{
        \node[font=\small] at (\cx,          -0.2) {\lbl};
        \node[font=\small] at (\cx+\xRshift, -0.2) {\lbl};
    }

    \foreach \cx in {1,2,3,4,5}{
        \draw[gridCol, line width=0.3pt]
        (\cx,          -\yI+0.4) -- (\cx,          -\yVI-0.4);
        \draw[gridCol, line width=0.3pt]
        (\cx+\xRshift, -\yI+0.4) -- (\cx+\xRshift, -\yVI-0.4);
    }

    \foreach \ry in {\yI,\yII,\yIII,\yIV,\yV,\yVI}{
        \draw[gridCol, line width=0.3pt]
        (0.5,-\ry) -- (5.5,-\ry);
        \draw[gridCol, line width=0.3pt]
        (0.5+\xRshift,-\ry) -- (5.5+\xRshift,-\ry);
    }

    \node[font=\small, rotate=90] at (\xEpL, -{(\yI+\yIII)/2})  {$\paredges{}=1$};
    \node[font=\small, rotate=90] at (\xEpL, -{(\yIV+\yVI)/2})  {$\paredges{}=5$};
    \node[font=\small, rotate=90] at (\xEpR, -{(\yI+\yIII)/2})  {$\paredges{}=1$};
    \node[font=\small, rotate=90] at (\xEpR, -{(\yIV+\yVI)/2})  {$\paredges{}=5$};

    \foreach \a/\ry in {
        0.3/\yI, 0.5/\yII, 0.7/\yIII,
        0.3/\yIV, 0.5/\yV, 0.7/\yVI}{
        \node[font=\footnotesize] at (\xAlpha,-\ry) {$\mcfalpha{=}\a$};
    }

    \draw[black!40, line width=0.4pt, dashed]
    (0, -\ySepA) -- (12, -\ySepA);

    \dcircle{1}{-\yI}{T}{TO}   \dcircle{2}{-\yI}{1}{1}   \dcircle{3}{-\yI}{17}{17}   \dcircle{4}{-\yI}{4}{4}   \dcircle{5}{-\yI}{33}{33}
    \dcircle{1}{-\yII}{T}{TO}  \dcircle{2}{-\yII}{1}{1}  \dcircle{3}{-\yII}{58}{58}  \dcircle{4}{-\yII}{3}{3}  \dcircle{5}{-\yII}{41}{41}
    \dcircle{1}{-\yIII}{T}{TO} \dcircle{2}{-\yIII}{1}{1} \dcircle{3}{-\yIII}{4}{4}   \dcircle{4}{-\yIII}{4}{4} \dcircle{5}{-\yIII}{116}{116}

    \dcircle{1}{-\yIV}{T}{TO}  \dcircle{2}{-\yIV}{1}{1}  \dcircle{3}{-\yIV}{63}{63}  \dcircle{4}{-\yIV}{4}{4}  \dcircle{5}{-\yIV}{44}{44}
    \dcircle{1}{-\yV}{T}{TO}   \dcircle{2}{-\yV}{1}{1}   \dcircle{3}{-\yV}{54}{54}   \dcircle{4}{-\yV}{4}{4}   \dcircle{5}{-\yV}{33}{33}
    \dcircle{1}{-\yVI}{T}{TO}  \dcircle{2}{-\yVI}{1}{1}  \dcircle{3}{-\yVI}{31}{31}  \dcircle{4}{-\yVI}{4}{4}  \dcircle{5}{-\yVI}{39}{39}

    \dcircle{1+\xRshift}{-\yI}{T}{TO}   \dcircle{2+\xRshift}{-\yI}{1}{1}   \dcircle{3+\xRshift}{-\yI}{T}{TO}    \dcircle{4+\xRshift}{-\yI}{2}{2}   \dcircle{5+\xRshift}{-\yI}{9}{9}
    \dcircle{1+\xRshift}{-\yII}{T}{TO}  \dcircle{2+\xRshift}{-\yII}{1}{1}  \dcircle{3+\xRshift}{-\yII}{179}{179} \dcircle{4+\xRshift}{-\yII}{2}{2}  \dcircle{5+\xRshift}{-\yII}{117}{117}
    \dcircle{1+\xRshift}{-\yIII}{T}{TO} \dcircle{2+\xRshift}{-\yIII}{1}{1} \dcircle{3+\xRshift}{-\yIII}{4}{4}   \dcircle{4+\xRshift}{-\yIII}{2}{2} \dcircle{5+\xRshift}{-\yIII}{118}{118}

    \dcircle{1+\xRshift}{-\yIV}{T}{TO}  \dcircle{2+\xRshift}{-\yIV}{1}{1}  \dcircle{3+\xRshift}{-\yIV}{T}{TO}   \dcircle{4+\xRshift}{-\yIV}{2}{2}  \dcircle{5+\xRshift}{-\yIV}{118}{118}
    \dcircle{1+\xRshift}{-\yV}{T}{TO}   \dcircle{2+\xRshift}{-\yV}{1}{1}   \dcircle{3+\xRshift}{-\yV}{T}{TO}    \dcircle{4+\xRshift}{-\yV}{2}{2}   \dcircle{5+\xRshift}{-\yV}{42}{42}
    \dcircle{1+\xRshift}{-\yVI}{T}{TO}  \dcircle{2+\xRshift}{-\yVI}{1}{1}  \dcircle{3+\xRshift}{-\yVI}{T}{TO}   \dcircle{4+\xRshift}{-\yVI}{2}{2}  \dcircle{5+\xRshift}{-\yVI}{9}{9}

\end{tikzpicture}\vspace{-10mm}
\caption{Median running times (in sec) over all \enquote{large} networks ($|V|\!>\!150$); \enquote{TO} denotes~timeout.}
\label{fig:running_times}
\end{figure}

\subparagraph{Success Rate and Running Time.} (cf.\ \Cref{fig:performance,fig:running_time_plots,fig:running_times})
A clear hierarchy of the algorithms %
emerges:
Due to its simplicity, %
\algspndone{} is almost always the fastest
algorithm, solving every instance in every parameter setting within one second.
The algorithm is closely followed by \algmcf{}, which only has to solve a single LP and thus always terminates on TZ instances in less than five (ten) seconds for the full-duplex (simplex, resp.) setting.
Even the largest DEFO and Rocketfuel topologies were usually solvable by \algmcf{} in the full-duplex setting whereas it reached the time limit otherwise. This is because \algmcf{}
was designed to leverage structural properties of full-duplex networks to speed up computations.
The costly post-processing provided by \algmcfpp{} operates more quickly in the simplex setting due to its greedy-like nature but still raises the running time significantly to several hundred seconds on large instances.

While the setting of $\paredges{}$ has nearly no effect on the running time of \algspndone{},
\algcps{} is highly sensitive to these parameters:
As it is designed for general digraphs, it has trouble with the full-duplex setting, where it already reaches the time limit for TZ instances with 50--100 vertices.
Unsurprisingly, \algcps{} is faster when considering only $\paredges{} = 1$ connection per link.
It is also faster for larger~$\mcfalpha$ (and indeed very quick for $\paredges{} = 1$, $\mcfalpha = 0.7$ in both communication settings), but degrades rapidly for smaller $\mcfalpha$. We speculate that this is due to simpler capacity interactions between different terminal pairs when more connections are active.

The most complicated algorithm \algspnd{} performs worst not only w.r.t.\ running time but indeed already w.r.t.\ the success rate, reaching the time limit of 10 minutes for 57\% of all TZ runs. Pilot studies showed that a larger time limit does not significantly improve the situation.
Generally, however, \algspnd{} can find provably optimal solutions for instances with up to roughly 80 vertices in both communication settings. %
Due to its size, the ILP is inherently difficult to solve even when using our column generation scheme. While the pricing routine is computationally cheap (only running for half a minute within the 10 minute time frame), a large amount of paths must be added to the ILP until a solution without violated dual path constraints is found.
Setting $\paredges{} = 5$ surprisingly lowers the running time,
with several TZ topologies containing 50--90 vertices only being solved for $\paredges{} = 5$ but not $\paredges{} = 1$. Presumably, a higher number of connections allows \algspnd{} to find good solutions that do not deactivate links completely, thus limiting the search of different possible routing paths.

\begin{figure}
    \centering
    \includegraphics[width=0.49\textwidth]{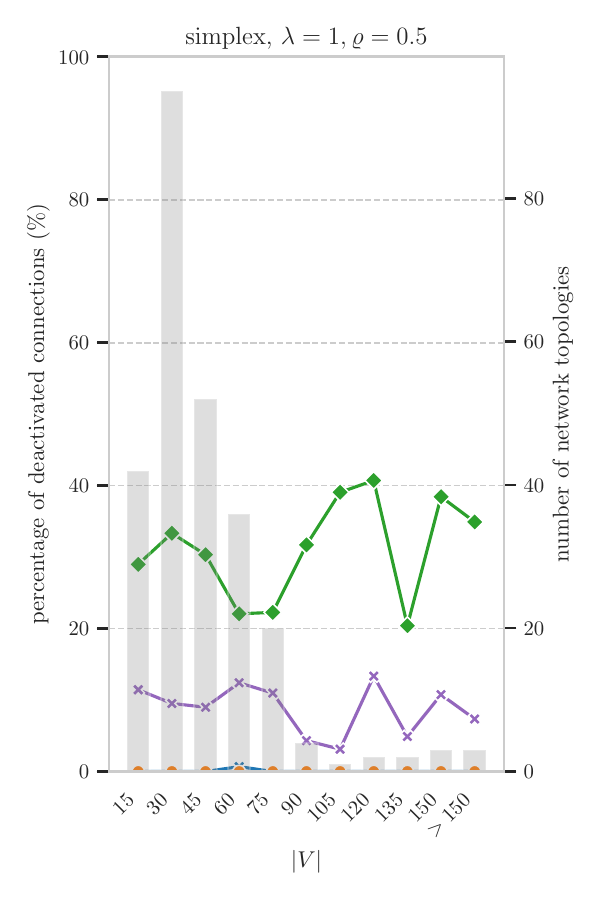}%
    \includegraphics[width=0.49\textwidth]{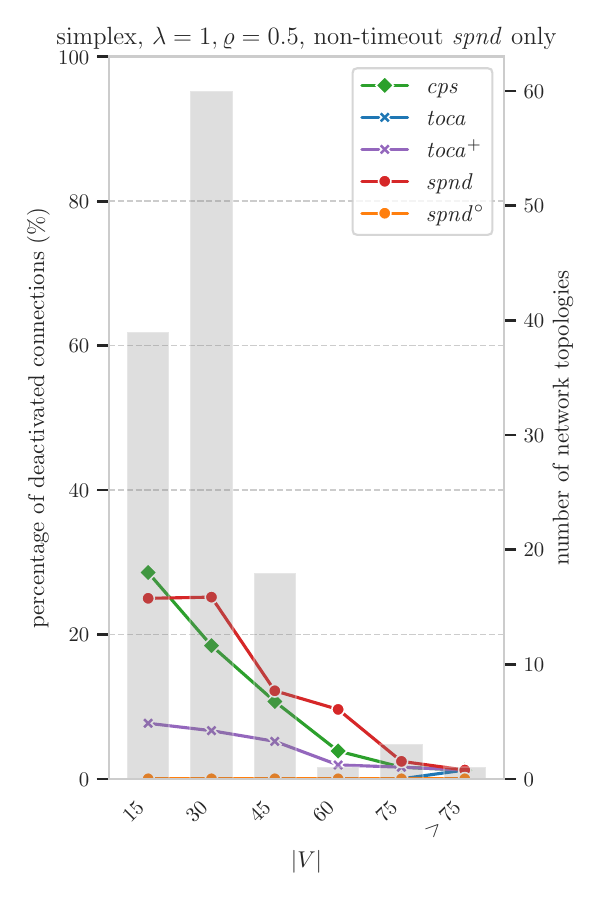}\\
    \includegraphics[width=0.49\textwidth]{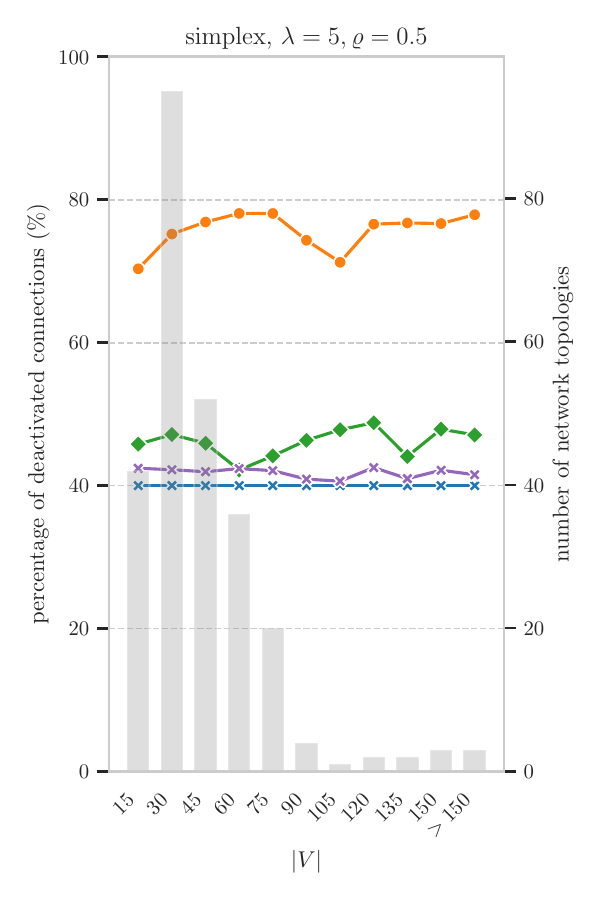}%
    \includegraphics[width=0.49\textwidth]{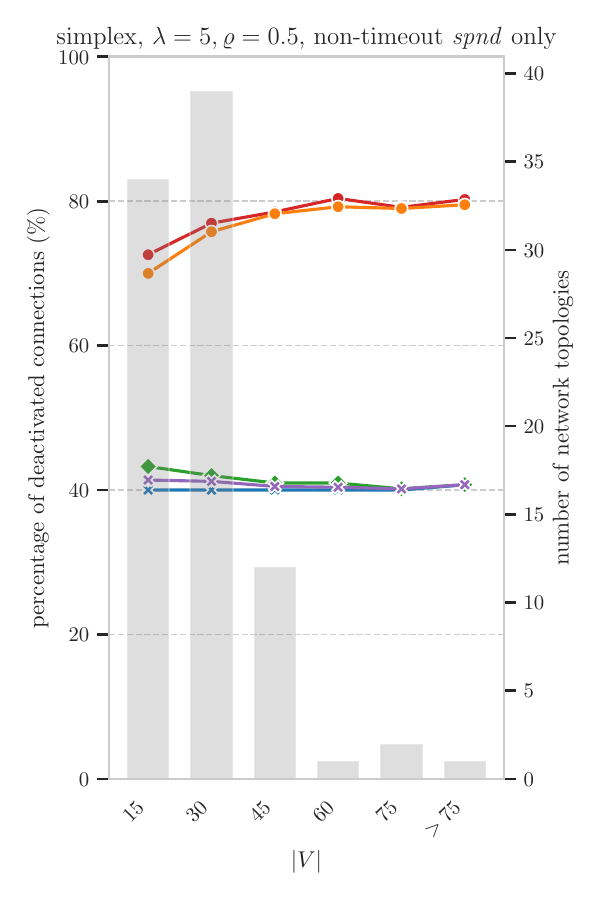}
    \caption{Median percentage of deactivated connections for representative parameter settings. On the left, \algspnd{} is not shown; on the right, only instances where \algspnd{} did not timeout are considered.}
    \label{fig:deactivated_connections}
\end{figure}

\begin{figure}
    \centering
    \includegraphics[width=0.49\textwidth]{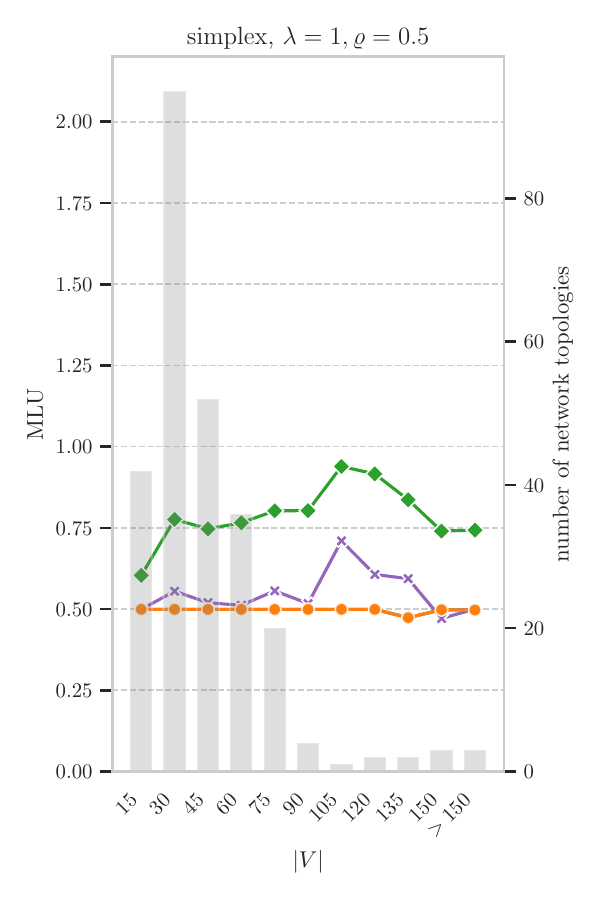}%
    \includegraphics[width=0.49\textwidth]{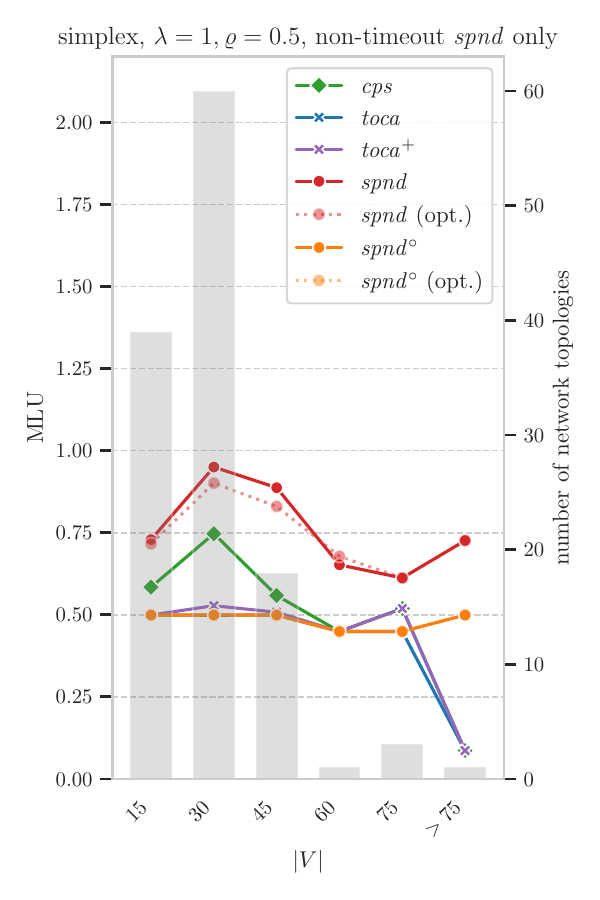}\\
    \includegraphics[width=0.49\textwidth]{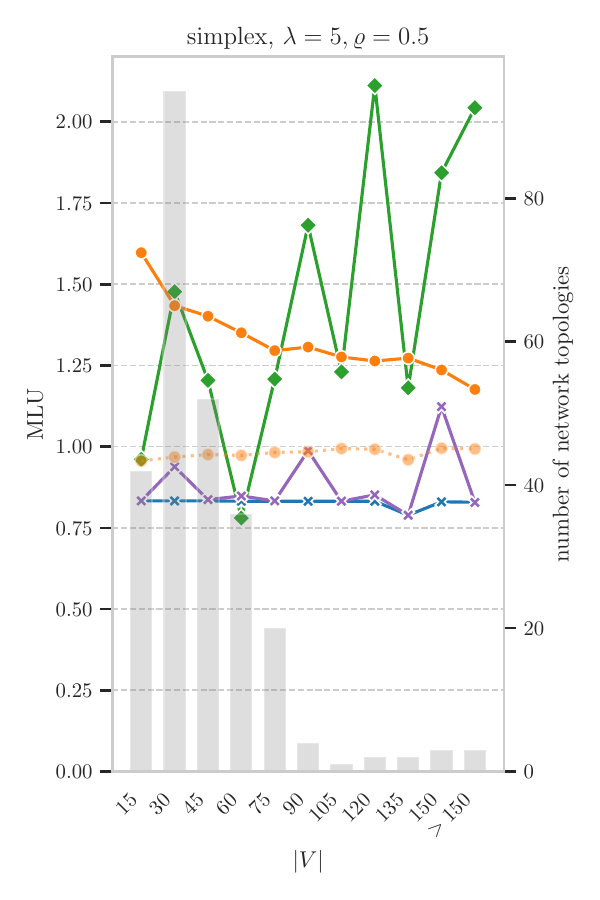}%
    \includegraphics[width=0.49\textwidth]{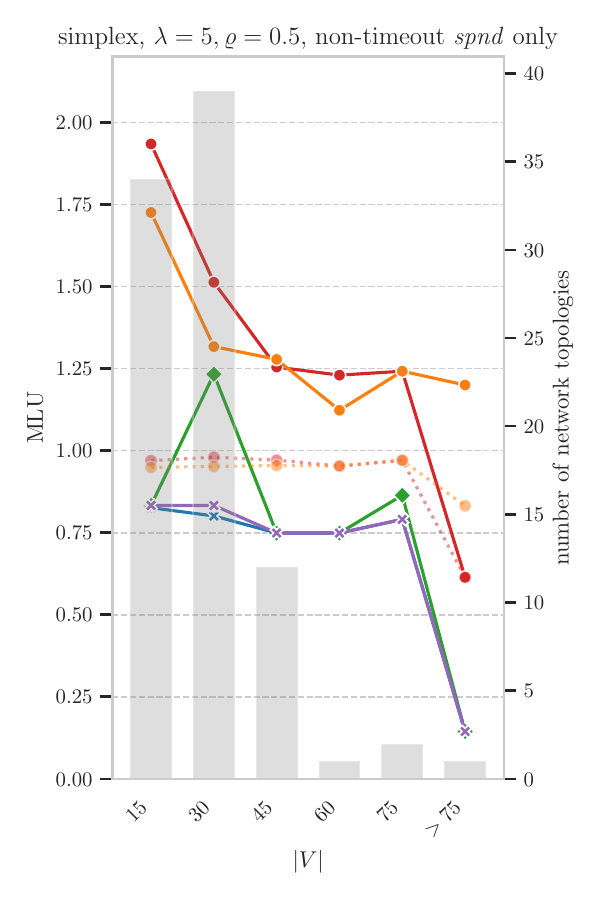}
    \caption{Median MLU for representative parameter settings.
     On the left, \algspnd{} is not shown; on the right, only instances where \algspnd{} did not timeout are considered.
    Dotted lines, labeled with \enquote{(opt.)}, show the MLU only w.r.t.\ to the traffic matrix that \algspnd{} or \algspndone{} was optimized for.}
    \label{fig:median_mlu}
\end{figure}

\subparagraph{Number of Active Connections (Solution Value).} (cf.\ \Cref{fig:deactivated_connections})
The great weakness of \algspndone{} becomes apparent for $\paredges{} = 1$,  where it is only able to deactivate those connections
that do not belong to links on any shortest path.
In practice, that means that there are only few, large instances where it can deactivate any connection at all.
In contrast, \algspnd{} shows the maximum number of connections that can be deactivated while still guaranteeing SPR-routability.
This value has a very high variance across instances,
depending greatly on the specific topology (recall that we only know this value for the comparably small instances due to \algspnd{}'s inapplicability to larger instances).
Clearly, we require more connections the larger the value of~$\mcfalpha$ is. The same holds when switching from simplex to full-duplex communication, as the deactivation of a connection requires that it is not needed in either direction.
In the simplex (full-duplex) setting, a maximum of 64\% (55\%) of connections can be deactivated for $\mcfalpha = 0.3$, and  53\% (41\%, resp.) for $\mcfalpha = 0.7$.

The traffic-oblivious algorithms \algcps{}, \algmcf{}, and \algmcfpp{} often surpass the upper bound given by \algspnd{}, in particular for larger instances and at lower $\mcfalpha$-values:
since they optimize for individual max-flows or MCF-routing respectively, they often return subnetworks that do not allow for a shortest-path routing anymore (see the \ac{MLU} comparison below).

The results shift greatly for $\paredges{} = 5$, which is the more realistic scenario since backbone networks contain multiple connections per link\cite{DBLP:journals/ojcs/OttenICA26}: Here, \algspndone{} suggests some of the highest
energy savings, deactivating only 1.1 percentage points fewer connections than \algspnd{}, which---depending on the communication setting and $\mcfalpha$---deactivates a median of 67--80\% of connections.
The traffic-oblivious algorithms deactivate much fewer connections since they aim to allow for a feasible routing even if the traffic changes.
Their reduction rates show very little variance in deactivated connections across network topologies, with \algmcf{} deactivating 60\% (40\%, 20\%) for $\mcfalpha=0.3$ ($0.5$, $0.7$, resp.), and \algmcfpp{} (\algcps{}) deactivating only a median of 2~(4, resp.) percentage points more connections.

\subparagraph{Maximum Link Utilization.} (cf.\ \Cref{fig:median_mlu})
It is only beneficial to deactivate connections if this does not inhibit the routability of the traffic.
Hence, we analyze the \acf{MLU}, i.e., the maximum ratio between the total traffic over an arc and its capacity in the computed subnetwork w.r.t.\ shortest-path routing. An MLU below 1.0 guarantees SPR-routability. Values above 1.0 suggest that one needs to consider a lower~$\mcfalpha$ compared to the actually measured traffic decrease (which is typically done for safety margins in practice anyhow~\cite{DBLP:journals/ton/SchullerACH21}).
Of course, \algspnd{} and \algspndone{} are optimized w.r.t.\ one specific traffic matrix~$\mcfalpha\demand$, allowing them to always obtain an MLU below 1.0 on~$\mcfalpha\demand$.
We show the median of these optimized MLUs as dotted lines.
To allow a fair comparison between the traffic-aware and -oblivious approaches, given a $\mcfalpha T$-specific solution by \algspnd{} and \algspndone{}, we also evaluate the median MLU across all five ($\mcfalpha$-scaled) traffic matrices for each instance.
Be aware that in rare occurrences, a traffic-aware algorithm may disconnect the network due to missing traffic between the connected components, leading to an infinite MLU on other traffic matrices.%

Naturally, the median MLU of the algorithms is strongly correlated to the number of connections they deactivate.
With increasing $\mcfalpha$, the median MLU also increases but showcases a lower variance across algorithms.
For %
$\paredges{} = 1$, the median MLU is also always below 1.0 for all algorithms, but the per-instance MLU for \algspnd{} and \algcps{} sometimes already surpass this threshold, suggesting that these algorithms deactivate connections too aggressively for traffic-oblivious scenarios.
This effect becomes more apparent for $\paredges{}=5$ (and the simplex setting in particular), where both traffic-aware algorithms and \algcps{} almost always surpass the MLU-threshold of 1, with \algspnd{} reaching higher MLU values than \algspndone{} (as it was more successful in deactivating links). The MLU of \algcps{} varies wildly, even reaching values above~2.0.
In contrast, \algmcf{} stays reliably in the 0.75--0.90 range and \algmcfpp{} only slightly above that, occasionally peaking above~1. %

\subparagraph{Key Takeaways.}
The traffic-oblivious algorithm with the best compromise between many deactivated connections and acceptable MLU is \algmcfpp{}.
In comparison, \algmcf{} serves as a useful alternative when a low running time is required. %
The algorithm~\algcps{} is clearly the wrong tool for the job as it is rather slow while taking neither the
simultaneous routing of all traffic demands nor the routing via shortest paths
into account, leading to extremely high MLUs.

While unsuited for general traffic-oblivious scenarios due to high MLUs, the trivial-to-implement and blazing fast \algspndone{} is the best choice for traffic-aware scenarios, which may include scenarios were the variance between different expected traffic matrices is very small. It deactivates close-to-optimal many connections and never surpasses an MLU of~1.

The exact algorithm \algspnd{} turned out to be a worthwhile tool: only by it, we were able to establish that \algspndone{} indeed finds solutions very close to the optimum. Yet, this comparably small amount of additionally deactivated connections by \algspnd{} does not justify its exorbitant running times
and effortful implementation,
so it does not lend itself to practical scenarios.

Overall, our evaluation gives credence to network design research that uses precomputed shortest paths in the input network as a basis for their routing~\cite{DBLP:conf/lcn/OttenBSA23}.

\section{Outlook}

Our proposed ILP for \prob{MSPND} uses the simplifying assumption that the shortest path for any~$(s,t)\in K$ is unique. Practical protocols like \ac{OSPF} deal with non-unique
shortest paths by implementing \ac{ECMP}, i.e., whenever shortest paths diverge, the traffic is split up equally among them.
Modeling this in a (still manageable) ILP is extremely difficult, so
computationally less expensive strategies to deal with this problem would be desirable. One could also ask for extending our model to
\prob{MSRND}, i.e., using a segment-routing protocol. However, the fact that our current model is already at the border of what is practically feasible to compute, this seems out of reach as of yet.

\newpage

\bibliography{main}

\begin{thebibliography}{10}

\bibitem{DBLP:conf/infocom/BhatiaHKL15}
Randeep Bhatia, Fang Hao, Murali~S. Kodialam, and T.~V. Lakshman.
\newblock Optimized network traffic engineering using segment routing.
\newblock In {\em Proc. {INFOCOM} 2015}, pages 657--665. {IEEE}, 2015.
\newblock \href {https://doi.org/10.1109/INFOCOM.2015.7218434} {\path{doi:10.1109/INFOCOM.2015.7218434}}.

\bibitem{DBLP:journals/cor/BourasFPZ19}
Ikram Bouras, Rosa Figueiredo, Michael Poss, and Fen Zhou.
\newblock On two new formulations for the fixed charge network design problem with shortest path constraints.
\newblock {\em Comput. Oper. Res.}, 108:226--237, 2019.
\newblock \href {https://doi.org/10.1016/J.COR.2019.04.007} {\path{doi:10.1016/J.COR.2019.04.007}}.

\bibitem{DBLP:journals/mmor/Braess68}
Dietrich Braess.
\newblock {{\"{U}}ber ein Paradoxon aus der Verkehrsplanung}.
\newblock {\em Unternehmensforschung}, 12(1):258--268, 1968.
\newblock \href {https://doi.org/10.1007/BF01918335} {\path{doi:10.1007/BF01918335}}.

\bibitem{DBLP:reference/crc/ChimaniGJKKM13}
Markus Chimani, Carsten Gutwenger, Michael J{\"{u}}nger, Gunnar~W. Klau, Karsten Klein, and Petra Mutzel.
\newblock {The Open Graph Drawing Framework (OGDF)}.
\newblock In {\em Handbook on Graph Drawing and Visualization}, pages 543--569. Chapman and Hall/CRC, 2013.
\newblock \href {https://doi.org/10.1201/b15385} {\path{doi:10.1201/b15385}}.

\bibitem{CPSJournal}
Markus Chimani and Max Ilsen.
\newblock Directed capacity-preserving subgraphs: hardness and exact polynomial algorithms.
\newblock {\em {Acta Informatica}}, 62(10), 2025.
\newblock \href {https://doi.org/10.1007/s00236-024-00475-7} {\path{doi:10.1007/s00236-024-00475-7}}.

\bibitem{DBLP:conf/isaac/ChimaniI25}
Markus Chimani and Max Ilsen.
\newblock Traffic-oblivious multi-commodity flow network design.
\newblock In {\em Proc.\ {ISAAC} 2025}, LIPIcs, pages 19:1--19:17. Schloss Dagstuhl - Leibniz-Zentrum f{\"{u}}r Informatik, 2025.
\newblock \href {https://doi.org/10.4230/LIPIcs.ISAAC.2025.19} {\path{doi:10.4230/LIPIcs.ISAAC.2025.19}}.

\bibitem{crainic2021networkdesign}
Teodor~Gabriel Crainic, Michel Gendreau, and Bernard Gendron, editors.
\newblock {\em Network Design with Applications to Transportation and Logistics}.
\newblock Springer International Publishing, 2021.
\newblock \href {https://doi.org/10.1007/978-3-030-64018-7} {\path{doi:10.1007/978-3-030-64018-7}}.

\bibitem{DBLP:journals/telsys/Dahl93}
Geir Dahl.
\newblock The design of survivable directed networks.
\newblock {\em Telecommun. Syst.}, 2(1):349--377, 1993.
\newblock \href {https://doi.org/10.1007/BF02109865} {\path{doi:10.1007/BF02109865}}.

\bibitem{DBLP:journals/dam/Dahl93}
Geir Dahl.
\newblock Directed steiner problems with connectivity constraints.
\newblock {\em Discret. Appl. Math.}, 47(2):109--128, 1993.
\newblock \href {https://doi.org/10.1016/0166-218X(93)90086-4} {\path{doi:10.1016/0166-218X(93)90086-4}}.

\bibitem{DantzigFulkerson}
George~Bernard Dantzig and Delbert~Ray Fulkerson.
\newblock {\em On the Max-Flow Min-Cut Theorem of Networks}, pages 215--222.
\newblock Princeton University Press, 2016.
\newblock \href {https://doi.org/10.1515/9781400881987-013} {\path{doi:10.1515/9781400881987-013}}.

\bibitem{DBLP:journals/nm/Dijkstra59}
Edsger~W. Dijkstra.
\newblock A note on two problems in connexion with graphs.
\newblock {\em Numerische Mathematik}, 1:269--271, 1959.
\newblock \href {https://doi.org/10.1007/BF01386390} {\path{doi:10.1007/BF01386390}}.

\bibitem{DBLP:journals/rfc/rfc8402}
Clarence Filsfils, Stefano Previdi, Les Ginsberg, Bruno Decraene, Stephane Litkowski, and Rob Shakir.
\newblock Segment routing architecture.
\newblock {\em {RFC}}, 8402:1--32, 2018.
\newblock \href {https://doi.org/10.17487/RFC8402} {\path{doi:10.17487/RFC8402}}.

\bibitem{DBLP:journals/corr/abs-1710-08665}
Steven Gay, Pierre Schaus, and Stefano Vissicchio.
\newblock {REPETITA:} repeatable experiments for performance evaluation of traffic-engineering algorithms.
\newblock {\em CoRR}, abs/1710.08665, 2017.
\newblock \href {https://arxiv.org/abs/1710.08665} {\path{arXiv:1710.08665}}.

\bibitem{DBLP:journals/jacm/GoldbergT88}
Andrew~V. Goldberg and Robert~Endre Tarjan.
\newblock A new approach to the maximum-flow problem.
\newblock {\em J. {ACM}}, 35(4):921--940, 1988.
\newblock \href {https://doi.org/10.1145/48014.61051} {\path{doi:10.1145/48014.61051}}.

\bibitem{DBLP:journals/cor/SilvaSMMS16}
Pedro~Henrique Gonz{\'{a}}lez, Luidi Simonetti, Philippe Michelon, Carlos~Alberto de~Jesus~Martinhon, and Edcarllos Santos.
\newblock A variable fixing heuristic with local branching for the fixed charge uncapacitated network design problem with user-optimal flow.
\newblock {\em Comput. Oper. Res.}, 76:134--146, 2016.
\newblock \href {https://doi.org/10.1016/J.COR.2016.06.016} {\path{doi:10.1016/J.COR.2016.06.016}}.

\bibitem{DBLP:conf/cp/HartertSVB15}
Renaud Hartert, Pierre Schaus, Stefano Vissicchio, and Olivier Bonaventure.
\newblock Solving segment routing problems with hybrid constraint programming techniques.
\newblock In {\em Proc. {CP} 2015}, volume 9255 of {\em LNCS}, pages 592--608. Springer, 2015.
\newblock \href {https://doi.org/10.1007/978-3-319-23219-5_41} {\path{doi:10.1007/978-3-319-23219-5_41}}.

\bibitem{DBLP:conf/sigcomm/HartertVSBFTF15}
Renaud Hartert, Stefano Vissicchio, Pierre Schaus, Olivier Bonaventure, Clarence Filsfils, Thomas Telkamp, and Pierre Fran{\c{c}}ois.
\newblock A declarative and expressive approach to control forwarding paths in carrier-grade networks.
\newblock In {\em Proc.\ {SIGCOMM} 2015}, pages 15--28. {ACM}, 2015.
\newblock \href {https://doi.org/10.1145/2785956.2787495} {\path{doi:10.1145/2785956.2787495}}.

\bibitem{DBLP:conf/infocom/HassidimRSS13}
Avinatan Hassidim, Danny Raz, Michal Segalov, and Ariel Shaqed.
\newblock Network utilization: The flow view.
\newblock In {\em Proc.\ {INFOCOM}}, pages 1429--1437. {IEEE}, 2013.
\newblock \href {https://doi.org/10.1109/INFCOM.2013.6566937} {\path{doi:10.1109/INFCOM.2013.6566937}}.

\bibitem{DBLP:journals/corr/abs-2601-13087}
Max Ilsen, Daniel Otten, Nils Aschenbruck, and Markus Chimani.
\newblock No traffic to cry: Traffic-oblivious link deactivation for green traffic engineering.
\newblock {\em CoRR}, abs/2601.13087, 2026.
\newblock Accepted at INFOCOM 2026.
\newblock \href {https://arxiv.org/abs/2601.13087} {\path{arXiv:2601.13087}}.

\bibitem{DBLP:books/daglib/0023873}
Michael J{\"{u}}nger, Thomas~M. Liebling, Denis Naddef, George~L. Nemhauser, William~R. Pulleyblank, Gerhard Reinelt, Giovanni Rinaldi, and Laurence~A. Wolsey, editors.
\newblock {\em 50 Years of Integer Programming 1958-2008 - From the Early Years to the State-of-the-Art}.
\newblock Springer, 2010.
\newblock \href {https://doi.org/10.1007/978-3-540-68279-0} {\path{doi:10.1007/978-3-540-68279-0}}.

\bibitem{DBLP:journals/jsac/KnightNFBR11}
Simon Knight, Hung~X. Nguyen, Nick Falkner, Rhys~Alistair Bowden, and Matthew Roughan.
\newblock The internet topology zoo.
\newblock {\em {IEEE} J. Sel. Areas Commun.}, 29(9):1765--1775, 2011.
\newblock \href {https://doi.org/10.1109/JSAC.2011.111002} {\path{doi:10.1109/JSAC.2011.111002}}.

\bibitem{DBLP:journals/rfc/rfc2328}
John Moy.
\newblock {OSPF} version 2.
\newblock {\em {RFC}}, 2328:1--244, 1998.
\newblock \href {https://doi.org/10.17487/RFC2328} {\path{doi:10.17487/RFC2328}}.

\bibitem{DBLP:conf/lcn/OttenBSA23}
Daniel Otten, Alexander Brundiers, Timmy Sch{\"{u}}ller, and Nils Aschenbruck.
\newblock Green segment routing for improved sustainability of backbone networks.
\newblock In {\em Proc.\ {LCN} 2023}, pages 1--9. {IEEE}, 2023.
\newblock \href {https://doi.org/10.1109/LCN58197.2023.10223317} {\path{doi:10.1109/LCN58197.2023.10223317}}.

\bibitem{DBLP:journals/ojcs/OttenICA26}
Daniel Otten, Max Ilsen, Markus Chimani, and Nils Aschenbruck.
\newblock An extended look at green traffic engineering by minimizing active linecards.
\newblock {\em {IEEE} Open J. Commun. Soc.}, 7:2794--2813, 2026.
\newblock \href {https://doi.org/10.1109/OJCOMS.2026.3673375} {\path{doi:10.1109/OJCOMS.2026.3673375}}.

\bibitem{DBLP:journals/ton/SchullerACH21}
Timmy Sch{\"{u}}ller, Nils Aschenbruck, Markus Chimani, and Martin Horneffer.
\newblock Failure resiliency with only a few tunnels - enabling segment routing for traffic engineering.
\newblock {\em {IEEE/ACM} Trans. Netw.}, 29(1):262--274, 2021.
\newblock \href {https://doi.org/10.1109/TNET.2020.3030543} {\path{doi:10.1109/TNET.2020.3030543}}.

\bibitem{DBLP:conf/lcn/SchullerACHS17}
Timmy Sch{\"{u}}ller, Nils Aschenbruck, Markus Chimani, Martin Horneffer, and Stefan Schnitter.
\newblock Predictive traffic engineering with 2-segment routing considering requirements of a carrier {IP} network.
\newblock In {\em Proc.\ {LCN}}, pages 667--675, 2017.

\bibitem{DBLP:conf/esa/SigurdZ04}
Mikkel Sigurd and Martin Zachariasen.
\newblock Construction of minimum-weight spanners.
\newblock In {\em Proc.\ {ESA} 2004}, Lecture Notes in Computer Science, pages 797--808. Springer, 2004.
\newblock \href {https://doi.org/10.1007/978-3-540-30140-0_70} {\path{doi:10.1007/978-3-540-30140-0_70}}.

\bibitem{DBLP:journals/ton/SpringMWA04}
Neil~T. Spring, Ratul Mahajan, David Wetherall, and Thomas~E. Anderson.
\newblock Measuring {ISP} topologies with rocketfuel.
\newblock {\em {IEEE/ACM} Trans. Netw.}, 12(1):2--16, 2004.
\newblock \href {https://doi.org/10.1109/TNET.2003.822655} {\path{doi:10.1109/TNET.2003.822655}}.

\end{thebibliography}

\end{document}